\pdfoutput=1
\documentclass[acmsmall]{acmart}
\makeatletter
\@ifundefined{KV@Gin@artifact}{\define@key{Gin}{artifact}[true]{}}{}
\makeatother

\usepackage{amsmath, amssymb}
\usepackage{dsfont}
\usepackage{listings}
\usepackage{xcolor}
\usepackage{cleveref}
\usepackage{tikz-cd}
\usepackage{mathpartir}
\lstdefinelanguage{agda}{
    keywords={record, data, =, $\forall$, where, eq/, squash/, _/_, [_], $\lambda$ $\equiv$, coinductive},
    basicstyle=\linespread{0.87}\footnotesize\ttfamily,
    aboveskip=3pt plus 1pt minus 1pt,
    belowskip=3pt plus 1pt minus 1pt,
    keywordstyle=\color{blue},
    commentstyle=\itshape\rmfamily,
    showstringspaces=false,
    columns=flexible,
    breaklines=true,
    texcl=true,
    mathescape=true,
    tabsize=4,
    stringstyle=\color{brown},
    escapeinside={(@}{@)},
}

\usepackage{etoolbox}
\BeforeBeginEnvironment{lstlisting}{\par\noindent\begin{minipage}{\linewidth}}
\AfterEndEnvironment{lstlisting}{\end{minipage}\par}

\makeatletter
\g@addto@macro\normalsize{%
  \setlength\abovedisplayskip{3pt plus 2pt minus 2pt}%
  \setlength\belowdisplayskip{3pt plus 2pt minus 2pt}%
  \setlength\abovedisplayshortskip{1pt plus 1pt}%
  \setlength\belowdisplayshortskip{1pt plus 1pt}%
}
\def\thm@space@setup{%
  \thm@preskip=3pt plus 1pt minus 1pt
  \thm@postskip=\thm@preskip}
\makeatother

\AtBeginDocument{%
  }

\usepackage{thmtools}

\theoremstyle{plain}
\newtheorem{theorem}{Theorem}[section]
\declaretheorem[sibling=theorem]{lemma}
\declaretheorem[sibling=theorem]{proposition}

\theoremstyle{definition}
\declaretheorem[sibling=theorem]{definition}
\declaretheorem[sibling=theorem]{example}

\theoremstyle{remark}
\declaretheorem[sibling=theorem]{remark}

\crefname{theorem}{theorem}{theorems}
\Crefname{theorem}{Theorem}{Theorems}

\crefname{lemma}{lemma}{lemmas}
\Crefname{lemma}{Lemma}{Lemmas}

\crefname{proposition}{proposition}{propositions}
\Crefname{proposition}{Proposition}{Propositions}

\crefname{corollary}{corollary}{corollaries}
\Crefname{corollary}{Corollary}{Corollaries}

\crefname{definition}{definition}{definitions}
\Crefname{definition}{Definition}{Definitions}

\crefname{example}{example}{examples}
\Crefname{example}{Example}{Examples}

\crefname{remark}{remark}{remarks}
\Crefname{remark}{Remark}{Remarks}

\newcommand{\outext}[2]{\text{oute}}

\newcommand{\outpart}[1]{\text{outp}}

\newcommand{\fst}[1]{\mathsf{fst}\,#1}

\DeclareMathOperator{\unit}{\mathds{1}}
\DeclareMathOperator{\unitt}{\ast}

\newcommand{\Type}{\mathsf{Type}}

\makeatletter
\newcommand{\proofpart}[2]{%
  \par
  \noindent\emph{Part #1: #2}\par\nobreak
  \@afterheading
}
\makeatother

\newcommand{\defineeq}{=_{\text{def}}}

\providecommand{\agdaurl}{https://prt-ema.pages.dev/agda}

\newcommand{\agdalogo}{%
  {\sffamily\bfseries\ensuremath{\checkmark}}%
}

\newcommand{\agda}[2]{\,\href{\agdaurl/#1.html\##2}{\agdalogo}}
\newcommand{\agdamod}[1]{\,\href{\agdaurl/#1.html}{\agdalogo}}

\setcopyright{cc}
\setcctype{by}
\acmDOI{10.1145/3871724}
\acmYear{2027}
\acmJournal{PACMPL}
\acmVolume{11}
\acmNumber{POPL}
\acmArticle{22}
\acmMonth{1}
\acmSubmissionID{popl27main-p193-p}
\received{2026-07-09}
\received[accepted]{2026-10-05}

\begin{document}

\title{Proof Repair and (Co)algebra Equivalences}

\author{E.G. Plante III}
\correspondingauthor
\orcid{0009-0008-8121-4292}
\affiliation{%
  \institution{University of Illinois Urbana-Champaign}
  \department{Siebel School of Computing and Data Science}
  \city{Urbana}
  \country{USA}
}
\email{eplante3@illinois.edu}

\author{Eric Paul}
\orcid{0009-0002-2880-5372}
\affiliation{%
  \institution{University of Illinois Urbana-Champaign}
  \city{Urbana}
  \country{USA}
}
\email{epaul9@illinois.edu}

\author{Cosmo Viola}
\orcid{0009-0009-1263-985X}
\affiliation{%
  \institution{University of Illinois Urbana-Champaign}
  \city{Urbana}
  \country{USA}
}
\email{cosmov2@illinois.edu}

\author{G. A. Kavvos}
\orcid{0000-0001-7953-7975}
\affiliation{%
  \institution{University of Bristol}
  \city{Bristol}
  \country{United Kingdom}
}
\email{alex.kavvos@bristol.ac.uk}

\author{Talia Ringer}
\orcid{0000-0003-1854-3321}
\affiliation{%
  \institution{University of Illinois Urbana-Champaign}
  \city{Urbana}
  \country{USA}
}
\email{tringer@illinois.edu}

\begin{abstract}
Maintaining formal proofs in response to breaking changes in their dependencies 
is a major challenge of proof engineering. \emph{Proof repair} is a kind of proof automation that lessens the burden of this challenge by automatically adapting proofs in response to those changes.
While proof repair is well understood through an algorithmic lens, it has, until now, been poorly understood
through a theoretical lens.
We address this gap by giving the first foundational treatment of proof repair,
focusing in particular on type-theoretic proof repair in response to changes in datatypes.
Our key insight is that the core component of a proof repair algorithm---a ``configuration''---is exactly an algebra isomorphism.
We formalize and prove this insight, answering open questions from prior work.
We then take this one step further and show that, by taking the dual of this representation, we obtain proof repair over coinductive types.
All of our results are formalized in Agda.
\end{abstract}



\NewDocumentCommand{\commentWrapper}{m}{\ifx\commentOut\undefined#1\else\fi}

\definecolor{commentblue}{RGB}{50,50,244}
\definecolor{commentpink}{RGB}{184,25,255}
\definecolor{commentgreen}{RGB}{0, 100, 0}

\newcommand\makenote[4]{\ifx\commentOut\undefined\par\noindent\begin{tcolorbox}
[size=small,colback=#2]\color{#3}\textbf{#4: }#1\end{tcolorbox}\noindent\else\fi}
\newcommand\talia[1]{{\color{commentgreen}\commentWrapper{(Talia: #1)}}}
\newcommand\todo[1]{{\color{commentblue}\commentWrapper{(TODO: #1)}}}
\newcommand\shit[1]{{\color{commentpink}\commentWrapper{[#1]}}}


\maketitle

\section{Introduction}
\label{sec:intro}

Formal proof developments in proof assistants like Agda, Rocq, and Lean~\cite{agda,rocq, lean} have grown larger as the technology has grown more accessible.
One timeless challenge in large formal proof developments is the burden of change~\cite{qedatlarge}.
Proof repair has emerged as a promising tool to cope with change by automatically fixing formal proofs in response to breaking changes.

This paper is about a class of tools for coping with changes in datatypes in proof developments that we refer to as \emph{type-theoretic proof repair}~\cite{ringer2021pldi,ringer_proof_2021,viola2025oopsla}: given an old and a new version of a datatype, these tools automatically rewrite the programs and proofs stated over the old version into programs and proofs stated over the new one.
Type-theoretic proof repair leverages the underlying type theory of a proof assistant to establish relationships between datatypes across which programs and proofs can be repaired.
This is, of course, not the only kind of proof repair---other repair tools exist that, for example, repair proofs concerning higher-order imperative code~\cite{gopinathan_mostly_2023}, or leverage non-type-theoretic methods like large language models~\cite{palm} (see \Cref{sec: related} for more).

The objective of this paper is to provide a mathematical foundation for type-theoretic proof repair tools that clarifies the scope of changes contemporary tools can handle and the extensions future repair tools can support.
To this end, we develop the connection between proof repair and the categorical notion of algebras on endofunctors, from which results concerning the limits, scope, and extensions of repair follow directly.
Our treatment characterizes the objects at the heart of repair tools; the syntactic transformations the tools perform over terms are beyond our present scope.

The principal theme of our work is that \textbf{proof repair is an algebra isomorphism}.
We bring this conclusion to light by relating prior proof repair work to the type-theoretic understanding of datatypes; the goal is to connect these two fields and leverage type theory to obtain formal results about proof repair tools.
Once established, the connection between algebra and repair gives a mathematical explanation for why repair works, settles questions left open in prior work, and relaxes assumptions repair was thought to require.
It also points beyond current tools, to repair for types no existing tool can handle.
All of our results and examples are marked with a blue checkmark linking to the corresponding Agda formalization.
\subsubsection*{Contributions}
Concretely, our contributions are:
\begin{itemize}
  \item We give a formal account of proof repair within type theory, which lays the groundwork for its connection to algebras (\Cref{sec: config}).
  In doing so, we formally establish results about the fundamental constructs of repair tools that prior work had only conjectured informally.
  \item We demonstrate that proof repair can be understood through the
notion of initial algebras on a suitable functor (\Cref{sec: alg-repair}). This new viewpoint clarifies the
class of relationships between types that configuration-based repair captures, characterizing them exactly in terms of universal properties. We recast a
case study from the original work on proof repair in terms of algebras, and formalize the
connection to proof repair algorithms. Our main result asserts that configurations are algebra isomorphisms.
  \item We find that repair does not require univalence, as was previously thought, and so applies to type theories that are inconsistent with univalence. 
  This opens the door for new implementations of repair in proof assistants like Lean (\Cref{subsec: non-ua}).
  \item With our new categorical perspective on proof repair, we develop
the foundations for repair algorithms that handle a new class of types: coinductive types (\Cref{sec: coinductive}).
  
\end{itemize}

\subsection*{Motivation}

We introduce proof repair using a classic case study \agdamod{inductive-repair.examples.nat}.
Suppose that a proof engineer has previously developed a verified library using the \texttt{Nat} datatype. They want to change the library to use a more efficient binary representation, \texttt{Bin}:

\begin{minipage}{0.48\textwidth}
\begin{lstlisting}
data Nat : Type where
    zero : Nat
    succ : Nat -> Nat
\end{lstlisting}
\end{minipage}
\hfill
\begin{minipage}{0.48\textwidth}
\begin{lstlisting}
data Pos : Type  where
 one : Pos
 pb0 : Pos -> Pos
 pb1 : Pos -> Pos

data Bin : Type where
 bzero : Bin
 bpos  : Pos -> Bin
\end{lstlisting}
\end{minipage}
The \texttt{Bin} datatype has a constructor for zero, \texttt{bzero}, and a constructor for positive numbers, \texttt{bpos}, which according to the \texttt{Pos} constructors is either $1$, $2n$, or $2n+1$ for some positive number $n$, itself represented by a term of type \texttt{Pos}.
Proof repair is the study of algorithms that replace one of these representations with the other. 
To see what repair must do, consider these definitions:

\noindent
\begin{minipage}{0.43\textwidth}
\begin{lstlisting}
isZeroN : Nat -> Bool
isZeroN zero     = true
isZeroN (succ _) = false
\end{lstlisting}
\end{minipage}
\hfill
\begin{minipage}{0.55\textwidth}
\begin{lstlisting}
zUnit : (n : Nat) -> n + zero = n
zUnit zero     = refl
zUnit (succ m) = ap succ (zUnit m)
\end{lstlisting}
\end{minipage}
While making a datatype change in a verified library, a proof engineer may encounter many such programs and proofs, requiring broad changes to their development.
A simple way of achieving that is to compose each old term with a conversion function \texttt{toNat : Bin -> Nat}: 

\noindent
\begin{minipage}{0.43\textwidth}
\begin{lstlisting}
isZeroB-nat : Bin -> Bool
isZeroB-nat b = isZeroN (toNat b)
\end{lstlisting}
\end{minipage}
\hfill
\begin{minipage}{0.55\textwidth}
\begin{lstlisting}
zUnitB-nat : (b : Bin) -> toNat b + zero = toNat b
zUnitB-nat b = zUnit (toNat b)
\end{lstlisting}
\end{minipage}
However, this output is unsatisfying because the transferred terms still route every operation back through \texttt{Nat} via \texttt{toNat}, so the old representation remains a dependency of the library.

Repair algorithms aim to do better. Given a \emph{configuration} (\Cref{sec: config}), i.e., a decomposition of a type equivalence into a correspondence between the constructors and eliminators of \texttt{Nat} and \texttt{Bin}, repair specializes the old code to \texttt{Bin}, producing native definitions that never mention \texttt{Nat}:

\noindent
\begin{minipage}{0.43\textwidth}
\begin{lstlisting}
isZeroB : Bin -> Bool
isZeroB bzero    = true
isZeroB (bpos _) = false
\end{lstlisting}
\end{minipage}
\hfill
\begin{minipage}{0.55\textwidth}
\begin{lstlisting}
zUnitB : (b : Bin) -> b + bzero = b
zUnitB bzero    = refl
zUnitB (bpos p) = ap bpos (pUnit p)
\end{lstlisting}
\end{minipage}

This algorithmic view of proof repair is well understood~\cite{ringer2021pldi}.
Before developing the algebra perspective, we give a formal account of proof repair, both to connect it to our abstract development and to establish results about the mechanics of repair.
\subsection*{Type Theory, Notation, and Conventions}\label{sec: conventions}

We work within an intensional Martin-L\"of type theory with dependent products, dependent sums, and identity types, using Agda as a proxy for this type theory.
Our development uses the homotopy-theoretic reading of the identity type~\cite{hottbook}, but assumes no univalence principles.
The only axioms we assume on the inductive side are function extensionality together with its computation rule on \textsf{happly} and the corresponding uniqueness rule (every path of functions is the extensionality of its \textsf{happly}).
For the coinductive development (\Cref{sec: coinductive}) we work in Cubical Agda~\cite{vezzosi_cubical_agda}, whose path type supports the copattern-style reasoning that coinductive types demand; there, function extensionality is a theorem, and again we make no use of univalence (or of other cubical-specific axioms beyond the path type itself).

We write $=$ for the identity (path) type, $p \bullet q : a = c$ for the composition of paths $p : a = b$ and $q : b = c$, $\mathsf{tpt}\,P\,p : P\,a \to P\,b$ for transport along $p : a = b$, and $\defineeq$ for definitions. We refer to a term of type $\Type$ as a \emph{carrier}.
We write $(a : A) \to B\,a$ for dependent functions and freely use $\to$ for its non-dependent instance, with $g \circ f : A \to C$ the composite of functions $f : A \to B$ and $g : B \to C$.
We write $(a : A) \times B\,a$ for dependent sums and freely use $\times$ for its non-dependent instance, with projections $\mathsf{fst}$ and $\mathsf{snd}$. We write $\unit$ for the unit type and $\unitt : \unit$ for its inhabitant. 
A \emph{type equivalence} $f : A \simeq B$ is a bi-invertible map: a function $f : A \to B$ together with a left inverse and a right inverse.
A type is \emph{contractible} when it has a term to which every term is equal; it is a \emph{proposition} when any two of its terms are equal, and a \emph{set} when its identity types are all propositions.
We suppress universe levels throughout; the formalization is universe-polymorphic.
Finally, when one equation depends on another we display it as a plain equality rather than writing the transport or dependent path explicitly; the Agda statements carry the dependency in full.

\section{Proof Repair via Configurations}\label{sec: config}

In this section, we internalize the core machinery of repair algorithms, \emph{configurations}, in type theory. 
The formalization of proof repair serves as the nexus of the connection we draw between repair and (co)algebras, and simultaneously provides a source of the questions the rest of the paper seeks to answer.
The results developed here about configurations are primarily confirmations of facts that were thought to be true of configurations, but were not previously made precise or proven in a formal setting.
We begin by informally working through the execution of a proof repair algorithm, and then give a formal definition of configurations.

The prior work on tools for type-theoretic repair (PUMPKIN Pi~\cite{ringer2021pldi,viola2025oopsla})
develops the notion of a configuration to capture a decomposed version of an equivalence $f : A \simeq B$ at the level of the metalanguage.
The repair tool then operates on the syntax of the type theory, replacing terms of type $A$ with terms of type $B$ according to this specification.
A configuration for the motivating example of \Cref{sec:intro} names each type's constructors by \texttt{DepConstr} and its dependent eliminator by \texttt{DepElim}, whose computation rules are \texttt{Iota}.\footnote{
We omit the \texttt{Eta} data that prior work included because our formalization shows it is redundant.
}
Together these present \texttt{Nat} and \texttt{Bin} through a common interface that pairs each \texttt{Nat} constructor (resp. eliminator) with the \texttt{Bin} constructor (resp. eliminator) it corresponds to under the equivalence, with constructors numbered by the metalanguage indices \texttt{0} and \texttt{1}.
Since \texttt{Bin} has no successor constructor, the configuration pairs \texttt{succ} with \texttt{binSucc : Bin -> Bin}; configured constructors need not be the datatype's own.
\begin{lstlisting}
-- DepConstr: constructors, indexed by type and constructor number
DepConstr Nat 0 = zero      DepConstr Bin 0 = bzero
DepConstr Nat 1 = succ      DepConstr Bin 1 = binSucc

-- DepElim: a dependent eliminator for each type
DepElim Nat : (P : Nat -> Type) -> P (DepConstr Nat 0) -> ((n : Nat) -> P n
              -> P (DepConstr Nat 1 n)) -> (n : Nat) -> P n
DepElim Bin : (P : Bin -> Type) -> P (DepConstr Bin 0) -> ((b : Bin) -> P b 
              -> P (DepConstr Bin 1 b)) -> (b : Bin) -> P b

-- Iota: the eliminator computes on each constructor
Iota Nat 0 : DepElim Nat P pz ps (DepConstr Nat 0)   = pz
Iota Nat 1 : DepElim Nat P pz ps (DepConstr Nat 1 n) = ps n (DepElim Nat P pz ps n)
Iota Bin 0 : DepElim Bin P pz ps (DepConstr Bin 0)   = pz
Iota Bin 1 : DepElim Bin P pz ps (DepConstr Bin 1 b) = ps b (DepElim Bin P pz ps b)
\end{lstlisting}

Proof repair tools then use the decomposition of the source and target types' inductive structure to automatically repair code by replacing the original constructors and eliminators component-wise with the new ones.
The repair algorithm replaces every \texttt{DepElim Nat} by \texttt{DepElim Bin}, \texttt{DepConstr Nat} by \texttt{DepConstr Bin}, and \texttt{Iota Nat} by \texttt{Iota Bin}, while the remaining case data is carried over unchanged.
For example, the proof \texttt{zUnit} given below would be repaired into \texttt{zUnitB}:
\begin{minipage}[t]{0.49\textwidth}
\begin{lstlisting}
zUnit : (n : Nat) -> n + zero = n
zUnit = DepElim Nat ($\lambda$n -> n + zero = n)
  (Iota Nat 0)
  ($\lambda$n ih -> Iota Nat 1 $\bullet$ ap succ ih)
\end{lstlisting}
\end{minipage}
\hfill
\begin{minipage}[t]{0.49\textwidth}
\begin{lstlisting}
zUnitB : (b : Bin) -> b + bzero = b
zUnitB = DepElim Bin ($\lambda$b -> b + bzero = b)
  (Iota Bin 0)
  ($\lambda$b ih -> Iota Bin 1 $\bullet$ ap binSucc ih)
\end{lstlisting}
\end{minipage}

The basis of a configuration lies in the constructors of a type, which may take recursive and non-recursive arguments. We specify these constructors by a \emph{signature}, and call an interpretation of a signature on a carrier a \textit{constructor algebra}.
Our definitions are inspired by Dybjer~\cite{dybjer_inductive} and Chapman et al.~\cite{chapman_levitation}, both of which proved crucial in formalizing configurations.
\begin{definition}[Constructor Algebra \agda{inductive-repair.indspec}{ConstrAlgebra}]\label{def: constr-alg}
  To describe the argument list of a single constructor as a telescope, we define a \emph{constructor arity} type generated by:
  \[
    \begin{array}{l@{\;}l}
      \quad\mathsf{Done}   &: \mathsf{ConstrArity}\\
      \quad\mathsf{Nonrec} &: (A : \Type) \to (A \to \mathsf{ConstrArity}) \to \mathsf{ConstrArity}\\
      \quad\mathsf{Rec}    &: (I : \Type) \to \mathsf{ConstrArity} \to \mathsf{ConstrArity}.
    \end{array}
  \]
  Here $\mathsf{Done}$ closes the telescope, $\mathsf{Nonrec}\,A\,k$ records a non-recursive argument of type $A$ on which the remaining arity $k$ may depend, and $\mathsf{Rec}\,I\,\mathit{cs}$ records a batch of recursive arguments, one for each element of $I$. 
  Non-recursive and recursive arguments may be freely interleaved.

  We define the type of \textit{arguments} to a constructor, parametrized by a carrier $X$ standing for the inductive type itself:
  \[
    \begin{aligned}
      &\mathsf{Args} : \mathsf{ConstrArity} \to \Type \to \Type\\
      &\mathsf{Args}\,\mathsf{Done}\,X \defineeq \unit\\
      &\mathsf{Args}\,(\mathsf{Nonrec}\,A\,k)\,X \defineeq (a : A) \times \mathsf{Args}\,(k\,a)\,X\\
      &\mathsf{Args}\,(\mathsf{Rec}\,I\,\mathit{cs})\,X \defineeq (I \to X) \times \mathsf{Args}\,\mathit{cs}\,X.
    \end{aligned}
  \]
  We call an element $\mathit{as} : \mathsf{Args}\,\mathit{cs}\,X$ an \emph{argument list}: a nested tuple with one entry for each step of the telescope $\mathit{cs}$, where a non-recursive entry is a value $a : A$ and a recursive entry is a map $r : I \to X$ giving the subterm at each position. In the clauses below we destructure argument lists as $(a, \mathit{as})$ or $(r, \mathit{as})$, with $\mathit{as}$ the remaining entries.
  A \emph{signature} is a type of constructor labels together with an assignment of an arity to each label:
  \[ 
    \begin{aligned}
      &\mathsf{Signature} \defineeq (\mathsf{Op} : \Type) \times (\mathsf{Op} \to \mathsf{ConstrArity})
    \end{aligned}
  \]
  Finally, a \emph{constructor algebra} for a signature on a carrier $X$ interprets every label as a constructor taking its arguments into the carrier:
  \[
    \begin{aligned}
      &\mathsf{ConstrAlgebra} : \mathsf{Signature} \to \Type \to \Type\\
      &\mathsf{ConstrAlgebra}\,\mathsf{(Op,\,arity)}\,X \defineeq (o : \mathsf{Op}) \to \mathsf{Args}\,(\mathsf{arity}\,o)\,X \to X.
    \end{aligned}
  \]
\end{definition}

For a map $g : X \to Y$ and $\mathit{cs} : \mathsf{ConstrArity}$, we define $\mathsf{mapArgs}\,\mathit{cs}\,g : \mathsf{Args}\,\mathit{cs}\,X \to \mathsf{Args}\,\mathit{cs}\,Y$ by postcomposing each family of recursive subterms $r : I \to X$ with $g$ and leaving non-recursive data untouched:
\[
  \begin{aligned}
    &\mathsf{mapArgs}\,\mathsf{Done}\,g\,\mathit{as} \defineeq \mathit{as}\\
    &\mathsf{mapArgs}\,(\mathsf{Nonrec}\,A\,k)\,g\,(a, \mathit{as})\defineeq (a,\, \mathsf{mapArgs}\,(k\,a)\,g\,\mathit{as})\\
    &\mathsf{mapArgs}\,(\mathsf{Rec}\,I\,\mathit{cs})\,g\,(r, \mathit{as})\defineeq (g \circ r,\, \mathsf{mapArgs}\,\mathit{cs}\,g\,\mathit{as}).
  \end{aligned}
\]
This action is functorial: it preserves identities and composition, and it carries an equivalence $X \simeq Y$ on carriers to an equivalence $\mathsf{Args}\,\mathit{cs}\,X \simeq \mathsf{Args}\,\mathit{cs}\,Y$ on arguments.

A configuration captures the shared inductive structure of two types by equipping each with a constructor algebra for a common signature, an eliminator for that algebra, and its $\beta$ law.
\begin{definition}[Configuration \agda{inductive-repair.config}{Config}]
  \label{def: configuration}
  Fix a signature $\mathit{sig} = (\mathsf{Op}, \mathsf{arity})$, a carrier $C$, and a constructor algebra $\alpha : \mathsf{ConstrAlgebra}\,\mathit{sig}\,C$.
  Relative to a \emph{motive} $P : C \to \Type$, the \textit{inductive hypotheses} at an argument list, $\mathit{as} : \mathsf{Args}\,\mathit{cs}\,C$, provide a proof of $P$ at every recursive child, that is, at $r\,d$ for each recursive entry $r$ of $\mathit{as}$ and each position $d$. Formally,
  \[
    \begin{aligned}
      &\mathsf{IH} : (\mathit{cs} : \mathsf{ConstrArity}) \to (C \to \Type) \to \mathsf{Args}\,\mathit{cs}\,C \to \Type\\
      &\mathsf{IH}\,\mathsf{Done}\,P\,\mathit{as} \defineeq \unit\\
      &\mathsf{IH}\,(\mathsf{Nonrec}\,A\,k)\,P\,(a, \mathit{as}) \defineeq \mathsf{IH}\,(k\,a)\,P\,\mathit{as}\\
      &\mathsf{IH}\,(\mathsf{Rec}\,I\,\mathit{cs})\,P\,(r, \mathit{as}) \defineeq \big((d : I) \to P\,(r\,d)\big) \times \mathsf{IH}\,\mathit{cs}\,P\,\mathit{as}.
    \end{aligned}
  \]
  An \emph{induction principle} for $\alpha$ takes a motive $P$ together with, for each label $o$, a case sending arguments $\mathit{as}$ and their inductive hypotheses to a value of $P\,(\alpha\,o\,\mathit{as})$, and produces a dependent section:
  \begin{align*}
    \mathsf{Ind}\,\mathit{sig}\,\alpha \defineeq&\, (P : C \to \Type) \to \\
      &\Big((o : \mathsf{Op})\;(\mathit{as} : \mathsf{Args}\,(\mathsf{arity}\,o)\,C) \to
        \mathsf{IH}\,(\mathsf{arity}\,o)\,P\,\mathit{as} \to P\,(\alpha\,o\,\mathit{as})\Big) \to\\
      &(x : C) \to P\,x.
  \end{align*}
  Any dependent function $f : (x : C) \to P\,x$ yields inductive hypotheses at every argument list, by applying $f$ to each recursive child (see the full definition in the formalization \agda{inductive-repair.indspec}{mkIHs}):
  \[
    \mathsf{mkIH} : (\mathit{cs} : \mathsf{ConstrArity}) \to ((x : C) \to P\,x) \to (\mathit{as} : \mathsf{Args}\,\mathit{cs}\,C) \to \mathsf{IH}\,\mathit{cs}\,P\,\mathit{as}.
  \]
  An eliminator $\mathsf{ind} : \mathsf{Ind}\,\mathit{sig}\,\alpha$ satisfies the \emph{$\beta$ law} when it computes on constructors: for all $P$, cases, label $o$, and arguments $\mathit{as}$,
  \[
    \mathsf{ind}\,P\,\mathit{cases}\,(\alpha\,o\,\mathit{as}) \;=\; \mathit{cases}\,o\,\mathit{as}\,\big(\mathsf{mkIH}\,(\mathsf{arity}\,o)\,(\mathsf{ind}\,P\,\mathit{cases})\,\mathit{as}\big).
  \]
  We write $\mathsf{Beta}\,\mathit{sig}\,\mathsf{ind}$ for the type of proofs that $\mathsf{ind}$ satisfies the $\beta$ law.
  An \emph{inductive algebra} for $\mathit{sig}$ on $C$ packages a constructor algebra with an induction principle and its $\beta$ law (referred to as $\iota$ in prior proof repair work):
  \[
    \mathsf{IndAlg}\,\mathit{sig}\,C \defineeq
      ({\alpha : \mathsf{ConstrAlgebra}\,\mathit{sig}\,C}) \times 
      ({\mathsf{ind} : \mathsf{Ind}\,\mathit{sig}\,\alpha}) \times
      \mathsf{Beta}\,\mathit{sig}\,\mathsf{ind}.
  \]
  Finally, a \emph{configuration} for $\mathit{sig}$ between carriers $C$ and $D$ is a pair of inductive algebras for $\mathit{sig}$, one on each carrier:
  \[
    \mathsf{Config}\,\mathit{sig}\,C\,D \defineeq \mathsf{IndAlg}\,\mathit{sig}\,C \times \mathsf{IndAlg}\,\mathit{sig}\,D.
  \]
\end{definition}

This definition internalizes the \texttt{DepConstr}/\texttt{DepElim}/\texttt{Iota} data that PUMPKIN Pi externally maintains; the results that follow are theorems about this reconstruction, and we support its faithfulness by recovering the tool's correctness criteria (\Cref{prop: correctness}) and case studies from it.

Though our signatures are non-indexed, our formalization handles the repair of inductive families in a way that is consistent with the proof repair literature.
Our formalization includes the case study from Ringer et al.~\cite{ringer2021pldi} of repairing between $\mathtt{Vec}\,A\,n$ and $\mathtt{Fin}\,n \to A$, where the per-index signature $\mathit{sig}\,n$ has a single non-recursive operation taking the $n$ values \agdamod{inductive-repair.examples.finvec}.
Programs and proofs over vectors repair index-by-index, with the coherence conditions of \Cref{prop: correctness} instantiated at each $n$.
The restriction to non-indexed signatures thus reflects the class of changes contemporary repair tools perform, which does not include repairing types with differing index types.

A configuration can also satisfy correctness criteria~\cite{ringer2021pldi} ensuring that it behaves coherently with the equivalence it is intended to capture.
Our formal definition satisfies these criteria.

\begin{proposition}[Correctness Conditions \agda{inductive-repair.coherence}{Coherence}]\label{prop: correctness}
  Fix $(\mathcal{A}_C, \mathcal{A}_D) : \mathsf{Config}\,\mathit{sig}\,C\,D$ with constructor algebras $\alpha_C, \alpha_D$, eliminators $\mathsf{ind}_C, \mathsf{ind}_D$, and $\beta$ laws $\beta_C, \beta_D$.
  Let $\varphi : C \to D$ be the forward map of the equivalence induced by the configuration (\Cref{prop: le-config-equiv}). The configuration satisfies the three coherence conditions of PUMPKIN~Pi~\cite{ringer2021pldi}.
  \begin{enumerate}
    \item $\mathsf{constr\_ok}$. Constructors commute with $\varphi$: for every label $o$ and arguments $\mathit{as} : \mathsf{Args}\,(\mathsf{arity}\,o)\,C$,
      \[
        \varphi\,(\alpha_C\,o\,\mathit{as}) \;=\; \alpha_D\,o\,\big(\mathsf{mapArgs}\,(\mathsf{arity}\,o)\,\varphi\,\mathit{as}\big).
      \]
      (The formal statement is a heterogeneous strengthening, relating $C$-side and $D$-side arguments across the lifted argument equivalence; the form displayed here is the instance at $\varphi$-mapped arguments.)
    \item $\mathsf{iota\_ok}$. Each eliminator computes on its own constructors: this is exactly the $\beta$ law, $\beta_C$ for $\mathsf{ind}_C$ and $\beta_D$ for $\mathsf{ind}_D$.
    \item $\mathsf{elim\_ok}$. The eliminators commute with $\varphi$: for every motive $P : D \to \Type$ and cases $\mathit{cases}_D$, there are transported $C$-side cases $\varphi^\ast \mathit{cases}_D$ (delegating to $\mathit{cases}_D$ on the $\varphi$-mapped arguments and transporting back along $\mathsf{constr\_ok}$) such that, for every $x : C$,
      \[
        \mathsf{ind}_C\,(P \circ \varphi)\,(\varphi^\ast \mathit{cases}_D)\,x \;=\; \mathsf{ind}_D\,P\,\mathit{cases}_D\,(\varphi\,x).
      \]
  \end{enumerate}
\end{proposition}
\begin{proof}[Proof sketch]
  All three follow from the configuration's structure and function extensionality.
\end{proof}

We apply the equivalence explicitly rather than using univalence and transport as in Ringer et al.~\cite{ringer2021pldi}. The first condition, $\mathsf{constr\_ok}$, states that $\varphi$ \emph{respects the constructors}, sending a constructed value to the corresponding construction on the other side.
The remaining two ask nothing of $\varphi$ in isolation: $\mathsf{iota\_ok}$ is just each type computing on its own constructors, and $\mathsf{elim\_ok}$ is recovered from $\mathsf{constr\_ok}$ once each type is known to be generated by its constructors in the strong sense the eliminators capture.
We will see in \Cref{sec: alg-repair} that this is no accident: a structure-respecting map between two such canonically generated types must be an isomorphism, and naming this phenomenon precisely is what the algebraic reformulation buys.
This is the first hint of our central thesis, that proof repair is an algebra isomorphism.
 
Additionally, our formalization enables us to prove a conjecture about configurations stated in Viola et al.~\cite{viola2025oopsla}.
The original proof repair work included extra data for an $\eta$ principle in both the configuration and the correctness conditions; we confirm that this data is not needed, since it is derivable from the rest \agda{inductive-repair.indspec}{IndAlg.eta}.

Configurations were developed as an algorithm-friendly decomposition of type equivalences, which in a tool setting are externalized in some automation language.
Equipped with our formal definitions, we are able to verify that configurations do indeed correspond to type equivalences.

\begin{proposition}[Configurations and Inductive Equivalences are Equivalent \agda{inductive-repair.config}{configEquiv}]
  \label{prop: le-config-equiv}
  Fix a signature $\mathit{sig}$ and carriers $C, D$.
  In the presence of an inductive algebra on $C$, type equivalences $C \simeq D$ and configurations between $C$ and $D$ are equivalent:
      \[ (\mathsf{IndAlg}\,\mathit{sig}\,C \;\times\; (C \simeq D)) \;\simeq\; \mathsf{Config}\,\mathit{sig}\,C\,D. \]
\end{proposition}
\begin{proof}[Proof sketch]
    An inductive algebra on $C$ appears on both sides, so, fixing one, it suffices to give an equivalence $(C \simeq D) \simeq \mathsf{IndAlg}\,\mathit{sig}\,D$.
    The inductive algebra transports directly across the equivalence.
    Conversely, a pair of inductive algebras determines an equivalence through their folds. Each inductive algebra sends a constructor algebra $\gamma$ on a type $Y$ to a map $\mathsf{fold}_C\,\gamma : C \to Y$, the eliminator at the constant motive $Y$, characterized by its $\beta$ rule
    \[ \mathsf{fold}_C\,\gamma\,(\alpha_C\,o\,\mathit{as}) \;=\; \gamma\,o\,\big(\mathsf{mapArgs}\,(\mathsf{arity}\,o)\,(\mathsf{fold}_C\,\gamma)\,\mathit{as}\big). \]
    Writing $\varphi = \mathsf{fold}_C\,\alpha_D$ and $\psi = \mathsf{fold}_D\,\alpha_C$, the composite $\psi \circ \varphi$ is a fold of $\alpha_C$ into itself; so is the identity, and the fold into a fixed algebra is unique, so $\psi \circ \varphi = \mathsf{id}_C$, and symmetrically $\varphi \circ \psi = \mathsf{id}_D$.
    These constructions are mutually inverse.
\end{proof}

This is not a surprising result, but it confirms the intuition with which configurations were designed and serves as a benchmark upon which the rest of the paper improves. 
However, this result does not give a satisfying mathematical description of repair, since it simply transports half of the data of a configuration across an equivalence to obtain the other half. 
We show (\Cref{sec: alg-repair}) that P-algebra isomorphisms provide an elegant mathematical answer to the question of what transformations configurations capture, and from this answer we obtain results about possible extensions of repair algorithms and the limits of repair.

In the typical case of repair, a proof engineer has a natural choice of configuration in mind, but, in general, unnatural and useless configurations exist.
For example, a configuration could reuse the same type or trivially specify a single constructor and eliminator.
It is up to the proof engineer and repair tooling to ensure that configurations execute a useful form of repair.
We return to this point in \Cref{sec: alg-repair}, where the algebraic reformulation makes this precise.

Configurations are a useful framework for building proof repair algorithms, but they expose a high level of syntactic detail that is a burden when trying to reason about repair in general. 
The rest of the paper is dedicated to developing an abstract formulation of repair that is more amenable to reasoning about the limits and properties of repair. We proceed by connecting configurations to P-algebras, and then observing the results of this connection in the context of repair algorithms.

\section{Algebraic Foundations for Proof Repair}\label{sec: alg-repair}

We formulate repair in the language of algebras on polynomial endofunctors~\cite{awodey_book,Adámek_Milius_Moss_2025}, which we introduce from scratch for readers unfamiliar with them.
We first introduce P-algebras and the algebraic view of inductive datatypes, providing the necessary background. 
We then establish that configurations are a representation of algebra isomorphisms. Finally, we showcase a purely algebraic form of proof repair.

\subsection{Warming Up to Algebras}

We begin with a tour of P-algebras.
Note that we are working within type theory, and so the definitions we use have a type-theoretic flavor and are not the precise definitions one may see in a category theory textbook.
We nonetheless refer to these type-theoretic notions by their usual category-theoretic names. 

\begin{definition}[Polynomial Endofunctor \agda{palg}{P}]\label{def: poly-endo}
  The polynomial endofunctor takes a type $A$ of \emph{shapes} and a family $B : A \to \mathsf{Type}$ assigning to each shape its type of \emph{positions}:
  \begin{align*}
    &P : (A : \mathsf{Type}) \to (A \to \mathsf{Type}) \to \mathsf{Type} \to \mathsf{Type}\\
    &P\,A\,B\,X \;\defineeq\; ({a : A}) \times \big(B\,a \to X\big)
  \end{align*}
  An element of $P\,A\,B\,X$ is a shape $a$ together with a $B\,a$-indexed family of children drawn from $X$. Acting on $f : X \to Y$ by postcomposition on the children makes $P\,A\,B$ a functor.
  \begin{align*}
    &\mathsf{P\text{-}map} : (X \to Y) \to P\,A\,B\,X \to P\,A\,B\,Y\\
    &\mathsf{P\text{-}map}\,f\,(a , u) \;\defineeq\; (a , f \circ u)
  \end{align*}
  We call a pair $(A : \Type,\, B : A \to \Type)$ a \textit{polynomial}.
\end{definition}

\begin{definition}[P-algebra \agda{palg}{P-Alg}]\label{def: p-alg}
  Fix $A : \Type$ and $B : A \to \Type$.
  For the polynomial endofunctor $P\,A\,B$, a \emph{$P$-algebra} is a pair $(C, \sigma)$ of a carrier $C$ and a structure map $\sigma : P\,A\,B\,C \to C$.
  We define a P-algebra formally as follows:
  \begin{align*}
  &\mathsf{P\text{-}Alg}\,A\,B : \Type \\
  &\mathsf{P\text{-}Alg}\,A\,B \defineeq (C : \Type) \times (P\,A\,B\,C \to C)
  \end{align*}
\end{definition}

\begin{example}[Nat Algebra \agda{inductive-repair.examples.nat}{sup\%E2\%84\%95}]
  Consider the functor
  \[ P : \Type \to \Type \qquad P\,X \defineeq \unit + X .\]
  A $P$-algebra is a pair $(C, \sigma)$ with a structure map $\sigma : \unit +\, C \to C$.
  By the universal property of the coproduct, such a map is equivalently a point $z : C$ together with an endomap $s : C \to C$, which we bundle as $\sigma = [\,z\,,\,s\,]$.
  Thus, a $P$-algebra is exactly a type equipped with an interpretation of a \texttt{zero} constructor and a \texttt{succ} constructor.
  The type \texttt{Nat} carries the algebra $(\mathtt{Nat}, [\,\mathtt{zero}, \mathtt{succ}\,])$, reading the nullary shape as \texttt{zero} and the unary shape as \texttt{succ}.
  The payoff of the algebraic view is that the same functor has many carriers: \texttt{Bin} carries $(\mathtt{Bin}, [\,\mathtt{bzero}, \mathtt{binSucc}\,])$, a different interpretation of the same two shapes.
\end{example}

\begin{definition}[Algebra Homomorphism \agda{palg}{AlgHom}]
  Given a polynomial $(A, B)$ and two algebras $(C, \sigma)$ and $(D, \tau)$ for $P\,A\,B$, an algebra homomorphism from $(C, \sigma)$ to $(D, \tau)$ is a function $h : C \to D$ such that the following diagram commutes:
  \[
    \begin{tikzcd}
      P C \arrow{r}{P h} \arrow{d}{\sigma} & P D \arrow{d}{\tau} \\
      C \arrow{r}{h} & D
    \end{tikzcd}
  \]
  Formally, we define the type of algebra homomorphisms as the function space between the carriers paired with the proof of the above commutative diagram:  \begin{align*}
    &\mathsf{isAlgHom} : (X\,Y\,: \mathsf{P\text{-}Alg}\,A\,B) \to (\mathsf{Carrier}\, X \to \mathsf{Carrier}\,Y) \to \Type\\
    &\mathsf{isAlgHom}\,(C,\,\sigma)\,(D,\,\tau)\,h \defineeq h \circ \sigma = \tau \circ \mathsf{P\text{-}map}\,h\\[4pt]
    &\mathsf{AlgHom} : \mathsf{P\text{-}Alg}\,A\,B \to \mathsf{P\text{-}Alg}\,A\,B \to \Type\\
    &\mathsf{AlgHom}\,X\,Y \defineeq ({h\,:\,\mathsf{Carrier}\,X \to \mathsf{Carrier}\,Y}) \times (\mathsf{isAlgHom}\,X\,Y\,h)
  \end{align*}
  Algebra homomorphisms compose ($\mathsf{AlgHom\text{-}}\circ$), and every algebra carries the identity homomorphism $\mathsf{id\text{-}hom}$, whose carrier map is the identity function.
\end{definition}

\begin{example}[Nat Homomorphism \agda{inductive-repair.examples.nat}{fromNat-hom}]
  Homomorphisms between $P$-algebras are maps that preserve the constructors of the source algebra.
  The map $\mathtt{fromNat} : \mathtt{Nat} \to \mathtt{Bin}$ is a homomorphism:
  \[\begin{tikzcd}
    {\unit +\,\mathtt{Nat}} && {\unit +\,\mathtt{Bin}} \\
    {\mathtt{Nat}} && {\mathtt{Bin}}
    \arrow["{\unit +\,\mathtt{fromNat}}", from=1-1, to=1-3]
    \arrow["{[\,\mathtt{zero},\,\mathtt{succ}\,]}"', from=1-1, to=2-1]
    \arrow["{[\,\mathtt{bzero},\,\mathtt{binSucc}\,]}", from=1-3, to=2-3]
    \arrow["{\mathtt{fromNat}}"', from=2-1, to=2-3]
  \end{tikzcd}\]
  This commutation is precisely the $\mathsf{constr\_ok}$ condition of \Cref{prop: correctness}: a homomorphism is a map that respects the constructors.
\end{example}

Since the central claim of this paper is that proof repair is an algebra isomorphism, we next fix precisely what an isomorphism of algebras is.

\begin{definition}[Algebra Equivalence \agda{palg}{AlgEquiv}]\label{def: alg-equiv}
  Given a polynomial $(A, B)$ and algebras $X, Y : \mathsf{P\text{-}Alg}\,A\,B$, an algebra homomorphism $f : \mathsf{AlgHom}\,X\,Y$ is an \emph{algebra equivalence} when it has a right inverse and a left inverse:
  \begin{align*}
    &\mathsf{isAlgEquiv} : (X\,Y : \mathsf{P\text{-}Alg}\,A\,B) \to \mathsf{AlgHom}\,X\,Y \to \Type\\
    &\mathsf{isAlgEquiv}\,X\,Y\,f \defineeq
      \big((g : \mathsf{AlgHom}\,Y\,X) \times (\mathsf{AlgHom\text{-}}\!\circ\,f\,g = \mathsf{id\text{-}hom}_X)\big)\\
    &\qquad\qquad\qquad\qquad\quad\times \big((h : \mathsf{AlgHom}\,Y\,X) \times (\mathsf{AlgHom\text{-}}\!\circ\,h\,f = \mathsf{id\text{-}hom}_Y)\big)\\[4pt]
    &\mathsf{AlgEquiv} : \mathsf{P\text{-}Alg}\,A\,B \to \mathsf{P\text{-}Alg}\,A\,B \to \Type\\
    &\mathsf{AlgEquiv}\,X\,Y \defineeq (f : \mathsf{AlgHom}\,X\,Y) \times \mathsf{isAlgEquiv}\,X\,Y\,f
  \end{align*}
  Forgetting the homomorphism witnesses, the carrier maps of an algebra equivalence form a type equivalence $\mathsf{Carrier}\,X \simeq \mathsf{Carrier}\,Y$; an algebra equivalence is exactly a type equivalence whose two maps respect the algebra structure.
  We say \emph{algebra isomorphism} interchangeably.
\end{definition}

\begin{definition}[Homotopy-initial Algebra \agda{palg}{isInitAlg}]\label{def: h-init-alg}
  Given a polynomial $(A, B)$, an algebra $X : \mathsf{P\text{-}Alg}\,A\,B$ is \emph{homotopy-initial} if for every $Y : \mathsf{P\text{-}Alg}\,A\,B$ the type of homomorphisms $X \to Y$ is contractible ($\mathsf{isContr}$).
  Formally,
  \begin{align*}
    &\mathsf{isInitAlg} : \mathsf{P\text{-}Alg}\,A\,B \to \Type\\
    &\mathsf{isInitAlg}\, X \defineeq (Y : \mathsf{P\text{-}Alg}\,A\,B) \to \mathsf{isContr} (\mathsf{AlgHom}\,X\,Y)
  \end{align*}
  For brevity, we may refer to homotopy-initial algebras as initial algebras.
\end{definition}

Equipped with our P-algebra toolbox, we turn back to our motivating example and phrase it in terms of algebras.
In this example we are implicitly relying on the correspondence between homotopy-initial algebras and W-types established by Awodey, Gambino, and Sojakova~\cite{awodey_homotopy-initial_2017}.
Of course, we are using Agda's \texttt{data} declaration instead of W-types, but we are still able to put an algebra structure on the Agda datatypes as they capture the same inductive structure.

\begin{example}[Motivating Example with Initial Algebras \agda{inductive-repair.examples.nat}{recNat}]
  First, we obtain that \texttt{Nat} and \texttt{Bin} are initial algebras for the functor $P\,X \defineeq \unit + X$.
  \begin{lstlisting}
    NatInit : isInitAlg (Nat , [ zero , succ ])
    BinInit : isInitAlg (Bin , [ bzero , binSucc ])
  \end{lstlisting}
  The unique homomorphism out of the initial algebra is the recursor, and its dependent strengthening is the induction principle. 
  Both are derived from \texttt{isInitAlg} alone---the recursor as the underlying function of the unique homomorphism, and the eliminator from the section it determines on the total space.
  \begin{lstlisting}
    recNat : {D : Type} -> D -> (D -> D) -> Nat -> D
    recNat-zero : recNat z s zero     = z
    recNat-succ : recNat z s (succ n) = s (recNat z s n)
    $\vspace{-0.08cm}$
    indNat : (Q : Nat -> Type) $\hspace{-0.08cm}$ -> Q zero $\hspace{-0.08cm}$ -> ((n : Nat) -> Q n -> Q (succ n)) $\hspace{-0.08cm}$ -> (n : Nat) $\hspace{-0.08cm}$ -> Q n
    indNat-zero : indNat Q qz qs zero     = qz
    indNat-succ : indNat Q qz qs (succ n) = qs n (indNat Q qz qs n)
  \end{lstlisting}
  With this, we can write the programs from before entirely in terms of universal properties.
  Here addition is itself defined by the recursor on its first argument, $m + n \defineeq \mathtt{recNat}\,n\,\mathtt{succ}\,m$, so the equations of $+$ used below are the computation rules of \texttt{recNat}:
  \begin{lstlisting}
    isZeroN : Nat -> Bool
    isZeroN = recNat true ($\lambda$ _ -> false)
    $\vspace{-0.08cm}$
    zUnit : (n : Nat) -> n + zero = n
    zUnit = indNat ($\lambda$ n -> n + zero = n) $\hspace{-0.08cm}$ recNat-zero $\hspace{-0.08cm}$ ($\lambda$ n ih -> recNat-succ $\bullet$ ap succ ih)
  \end{lstlisting}
  Because \texttt{Bin} is initial for the same functor, it carries the same universal property, the same recursor and eliminator, now out of \texttt{Bin}. 
  Repair is then a simple instantiation of the universal property at the new initial algebra: $\mathtt{recNat}/\mathtt{indNat}$ become $\mathtt{recBin}/\mathtt{indBin}$ and $\mathtt{zero}/\mathtt{succ}$ become $\mathtt{bzero}/\mathtt{binSucc}$, while the programs' defining data is untouched.
  \begin{lstlisting}
    recBin : {D : Type} -> D -> (D -> D) -> Bin -> D
    indBin : (Q : Bin -> Type) $\hspace{-0.08cm}$ -> Q bzero $\hspace{-0.08cm}$ -> $\hspace{-0.08cm}$ ((b : Bin) $\hspace{-0.08cm}$-> $\hspace{-0.08cm}$Q b $\hspace{-0.08cm}$-> $\hspace{-0.08cm}$Q (binSucc b)) $\hspace{-0.09cm}$ -> $\hspace{-0.08cm}$ (b : Bin) $\hspace{-0.15cm}$ -> $\hspace{-0.15cm}$ Q $\hspace{-0.08cm}$b
    $\vspace{-0.08cm}$
    isZeroB : Bin -> Bool
    isZeroB = recBin true ($\lambda$ _ -> false)
    $\vspace{-0.08cm}$
    zUnitB : (b : Bin) -> b + bzero = b
    zUnitB = indBin $\hspace{-0.08cm}$ ($\lambda$ b -> b + bzero = b) $\hspace{-0.08cm}$ recBin-zero ($\lambda$ b ih -> recBin-succ $\bullet$ ap binSucc ih)
  \end{lstlisting}
\end{example}
The algebraic perspective on proof repair is epitomized by the preceding example; the two very different representations of natural numbers are initial for the same functor and hence must have a unique isomorphism between them.
This is our thesis emerging: \texttt{Nat} and \texttt{Bin} are repaired precisely through this algebra isomorphism. 
We spend the rest of this section cementing this idea and establishing that our slogan (proof repair is an algebra isomorphism) is general and concrete.
\subsection{Connecting Configurations and Algebras}

Type equivalences are an unsatisfactory answer to the question of what configurations capture.
Indeed, because type equivalences do not capture the inductive structure of configurations, our proof of \Cref{prop: le-config-equiv} assumed an inductive algebra to generate a configuration from an equivalence.
We use P-algebras to put our understanding of configurations on firmer ground.
The main idea of this section is that configurations are a syntax-sensitive method for capturing a shared inductive structure on types, which can be represented by algebra isomorphisms.
We begin to develop this idea by coaxing configurations into something less burdened by syntax.

Every configuration has a canonical form that we call a \textit{W-type configuration}. These configurations emulate W-types
by capturing the inductive structure of types with a single constructor that takes a non-recursive and a recursive argument.
\begin{definition}[W-type Configurations \agda{inductive-repair.wtype-config}{WTypeConfig}]\label{def: wtype-config}
  We define a one-constructor \emph{W-type signature}: a single operation whose arity supplies a shape $a : A$ and then a $B\,a$-indexed family of recursive children. 
  \begin{align*}
    &\mathsf{WTypeArity} : (A : \Type) \to (A \to \Type) \to \unit \to \mathsf{ConstrArity}\\
    &\mathsf{WTypeArity}\,A\,B\,\_ \defineeq \mathsf{Nonrec}\, A\, (\lambda a \to \mathsf{Rec}\, (B\, a)\, \mathsf{Done})\\[4pt]
    &\mathsf{WTypeSignature} : (A : \Type) \to (A \to \Type) \to \mathsf{Signature}\\
    &\mathsf{WTypeSignature}\, \_ \, \_ \, \mathsf{.Signature.Op} \defineeq \unit \\
    &\mathsf{WTypeSignature}\,A\,B\, \mathsf{.Signature.arity} \defineeq \mathsf{WTypeArity}\,A\,B\\[4pt]
    &\mathsf{WTypeIndAlg} : (A : \Type) \to (A \to \Type) \to (C : \Type) \to \Type\\
    &\mathsf{WTypeIndAlg}\,A\,B\,C \defineeq \mathsf{IndAlg}\, (\mathsf{WTypeSignature}\,A\,B)\,C
  \end{align*}
  A \emph{W-type configuration} is the special configuration where the signature is a W-type signature, collapsing all per-constructor data into one constructor with one non-recursive argument and one recursive argument:
  \begin{align*}
    &\mathsf{WTypeConfig} : (A : \Type) \to (A \to \Type) \to (C\,D : \Type) \to \Type\\
    &\mathsf{WTypeConfig}\,A\,B\,C\,D \defineeq \mathsf{Config}\, (\mathsf{WTypeSignature}\,A\,B)\,C\,D
  \end{align*}
\end{definition}

Whereas configurations are useful for capturing the syntactic detail of particular representations of inductive types,
W-type configurations collapse that detail, which makes them useful for reasoning about the underlying structure configurations capture.
Indeed, we prove every configuration is represented by a corresponding W-type configuration.

\begin{lemma}[Configurations Faithfully Embed into W-type Configurations \agdamod{inductive-repair.wtype-config}]\label{lem: embed-wtype}
  Every configuration over an arbitrary signature $\mathit{sig}$ determines a W-type configuration on the same carriers. We read off the polynomial whose shapes pair a constructor label with its non-recursive data (an argument list at the carrier $\unit$) and whose positions ($\mathsf{recPos}$) are the recursive arguments occurring in that shape:
  \[ A \;\defineeq\; {(c : \mathsf{Op})} \times \mathsf{Args}(\mathsf{arity}\,c)\,\unit, \qquad B\,(c, s) \;\defineeq\; \mathsf{recPos}(\mathsf{arity}\,c)\,s, \]
  and the configuration transports to a W-type configuration for this $(A, B)$:
  \begin{align*}
    \mathsf{Config}\,\mathit{sig}\,C\,D \to ({A : \Type}) \times (({B : (A \to \Type)}) \times (\mathsf{WTypeConfig}\,A\,B\,C\,D))
  \end{align*}
  The embedding is faithful: writing $\varphi, \psi : C \leftrightarrows D$ for the forward and backward maps of the repair equivalence induced by the configuration (\Cref{prop: le-config-equiv}), and $\varphi_W, \psi_W$ for those induced by its W-type configuration, the maps are equal:
  \[ \varphi_W = \varphi \qquad\text{and}\qquad \psi_W = \psi. \]
\end{lemma}

Thus, arbitrary configurations translate to W-type configurations with the repair maps preserved.
This canonicalization is essential, as it lets us transfer type-theoretic results about W-types to repair.
In particular, by showing that configurations over arbitrary signatures carry the same repair data as W-type configurations, we connect repair to the study of W-types by Awodey et al.~\cite{awodey_homotopy-initial_2017}.
Because nothing about repair is lost by passing to the canonical form, we state our main theorem over W-type configurations.

\begin{theorem}[W-type Configurations are Unique Algebra Equivalences \agdamod{inductive-repair.config-alg-equiv}]\label{thm: wtype-algequiv}
  For any polynomial $(A, B)$ and carriers $C, D$, a W-type configuration is exactly an algebra equivalence between algebras on $C$ and $D$ whose source is initial:
    \begin{align*}
      \mathsf{WTypeConfig}\,A\,B\,C\,D \;\simeq
        & (\sigma_C : P\,A\,B\,C \to C) \\
        \times\, & (\sigma_D : P\,A\,B\,D \to D) \\
        \times\, & \mathsf{AlgEquiv}\,(C, \sigma_C)\,(D, \sigma_D) \\
        \times\, & \mathsf{isInitAlg}\,(C, \sigma_C)
  \end{align*}
\end{theorem}

\begin{proof}[Proof sketch]
  The carrier-level equivalence $\mathsf{WTypeIndAlg}\,A\,B\,X \simeq \mathsf{InitAlgOn}\,A\,B\,X$ (with $\mathsf{InitAlgOn}$ as in \Cref{def: alg-config}) is given by the homotopy-initiality and W-type correspondence of Awodey et al.~\cite{awodey_homotopy-initial_2017}, which we then use to get an equivalence between the corresponding pairs.
  We then construct an equivalence between a pair of initial algebras and a single initial algebra together with an algebra equivalence to the other.
  Composing these equivalences yields the desired conclusion.
  That the equivalence is unique \agdamod{inductive-repair.alg-iso} is given by the universal property~\cite{awodey_homotopy-initial_2017}.
\end{proof}

This result gives a bridge between repair and algebras that justifies our slogan that repair is an algebra isomorphism.
Once configurations are translated into W-type configurations the theorem is unsurprising, but it is valuable for the conceptual connection it cements between repair and P-algebras.
Indeed, the algebraic perspective clarifies properties of repair that the purely syntactic view offered by configurations (\Cref{sec: config}) does not.
What is remarkable about this correspondence is that it shows how configurations realize algebra isomorphisms as syntactic transformations that remove old dependencies (\Cref{prop: fold-swap}), so that repair carries out a mathematically defined transformation without ever referring to the old type.

We now know that a configuration packages precisely the data of a universal construction on two carriers, namely the property of both carriers being initial for a common functor.
For a proof engineer, the question of when repair between two datatypes is natural can now be posed precisely and formally as whether there is a natural functor for which both types present an initial algebra.
Given a positive answer, the rest of repair is determined; the transformation the algorithm executes is exactly the algebra isomorphism that initiality forces to exist uniquely.

Additionally, we can state that configurations can handle exactly those datatype changes representable by an algebra isomorphism.
This class is broad in practice, since the initial algebra of a functor pins down the data structure the programmer has in mind, so changes of representation that preserve that structure fall within its reach.
This characterization has force only relative to a fixed polynomial, since any equivalence is an algebra isomorphism for a constant functor.

\begin{example}[Motivating Example with Algebra Equivalence \agda{inductive-repair.examples.nat}{nat-bin-eqv}]
    Because \texttt{Nat} and \texttt{Bin} are initial for the same functor, initiality hands us a homomorphism each way and forces their two composites to be endomorphisms of an initial algebra, hence the identities.
    This gives the unique algebra isomorphism $\mathtt{toNat/fromNat} : \mathtt{Nat} \simeq \mathtt{Bin}$.
    Now the swap of initial algebras from the previous example is the action of this isomorphism.
    The new recursor $\mathtt{recBin}\,z\,s$ precomposed with $\mathtt{fromNat}$ is again a homomorphism, so by initiality of \texttt{Nat} it must coincide with $\mathtt{recNat}\,z\,s$; the swap is forced, not chosen.
    Each old program thus equals the repaired one precomposed with the isomorphism:
    \begin{lstlisting}
      recNatSwap : recNat z s = recBin z s $\circ$ fromNat 
      $\vspace{-0.08cm}$
      isZeroCoh : (n : Nat) -> isZeroN n = isZeroB (fromNat n)
      isZeroCoh n = happly recNatSwap n             
    \end{lstlisting}
    The same story plays out one level up for proofs: the dependent eliminators commute with the isomorphism ($\mathsf{elim\_ok}$), so the repaired proof transports back to the original along the round-trip witness $\mathtt{toNatFromNat}$.
    Writing $\mathtt{zUnitB\text{-}nat}$ for the composed \texttt{Bin}-side proof from the introduction (ranging over $\mathtt{toNat}\,b$, so the predicate stays on the \texttt{Nat} side), we have
    \begin{lstlisting}
      zUnitB-nat : (b : Bin) -> toNat b + zero = toNat b
      zUnitCoh : (n : Nat) $\hspace{-0.08cm}$-> $\hspace{-0.08cm}$tpt $\hspace{-0.08cm}$($\lambda$ m -> m + zero = m) $\hspace{-0.08cm}$(toNatFromNat n) $\hspace{-0.08cm}$(zUnitB-nat (fromNat n)) $\hspace{-0.08cm}$= $\hspace{-0.08cm}$zUnit $\hspace{-0.08cm}$n
    \end{lstlisting}
    Repair across \texttt{Nat} and \texttt{Bin} is, in this precise sense, the unique algebra isomorphism between them: the proof engineer's swap of constructors and eliminators is exactly its action on terms, and the coherence conditions of~\Cref{prop: correctness} are precisely the witnesses that this action sends each original term to its repair.
    However, the naive use of an algebra isomorphism is not repair since it does not rid the code of the old dependency.
    Shortly, we characterize the ability to get rid of old datatypes in terms of P-algebras.
\end{example}
\subsection{Algebraic Repair}
Now, we can formulate repair entirely in P-algebraic terms.
\begin{definition}[Algebra Configuration \agda{inductive-repair.palg-config}{AlgConfig}]\label{def: alg-config}
  Given a polynomial $(A, B)$ and carriers $C, D$, an \emph{initial algebra structure} on $C$ packages an algebra with a proof of initiality:
  \[ \mathsf{InitAlgOn}\,A\,B\,C \;\defineeq\; (\sigma : P\,A\,B\,C \to C) \times \mathsf{isInitAlg}\,(C, \sigma). \]
  An \emph{algebra configuration} between $C$ and $D$ is a pair of initial algebra structures, one on each carrier:
  \[ \mathsf{AlgConfig}\,A\,B\,C\,D \;\defineeq\; \mathsf{InitAlgOn}\,A\,B\,C \;\times\; \mathsf{InitAlgOn}\,A\,B\,D. \]
\end{definition}

An algebra configuration is deliberately the algebraic analogue of a configuration (\Cref{def: configuration}): where a configuration carries an inductive algebra on each carrier, an algebra configuration carries an initial P-algebra on each carrier. 

An algebra configuration determines an algebra isomorphism between its two carriers by the initiality of each side, which we refer to as the repair map.
Forgetting the homomorphism witnesses and keeping the underlying carrier maps recovers exactly the type equivalence $C \simeq D$ of~\Cref{prop: le-config-equiv}; the algebraic statement is the upgrade, an equivalence whose maps preserve the inductive structure rather than a bare bijection of carriers.

This algebraic reformulation is more than a change of vocabulary.
The algebraic picture replaces the syntax-heavy bookkeeping of a configuration with a single mathematical object, a pair of initial algebras for one polynomial functor.
This puts the foundations of proof repair on well-trodden mathematical ground, and clarifies what repair computes. 
It equally clarifies the proof engineer's task.
The engineer needs only the common algebraic structure shared by the source and target; the repair, together with its correctness, is then determined.
The payoff of the slogan that proof repair is an algebra isomorphism is that it names the goal of a repair algorithm precisely and grounds it in a mathematical fact independent of any particular implementation.

The coherence conditions are subsumed by this viewpoint.
In the configuration setting of~\Cref{prop: correctness}, the three conditions are properties one must establish of the induced map. From an algebra configuration they fall out of initiality, and our development derives all three from the two initial algebras alone \agda{inductive-repair.palg-coherence}{AlgCoherence}.
Commutation with the constructors ($\mathsf{constr\_ok}$) is given by the homomorphism square and the remaining two rest on the source being canonically generated by its constructors, which for an initial algebra is exactly its dependent eliminator and the eliminator's computation rule.
This is the precise content of the foreshadowing in~\Cref{sec: config}.

\begin{proposition}[Algebra Config iff Configuration \agda{inductive-repair.conversion}{configToPalg}]\label{prop: alg-config-iff}
  Over carriers $C, D$, write
  \begin{align*}
    &\mathsf{SigConfig}\,C\,D \defineeq (\mathit{sig} : \mathsf{Signature}) \times \mathsf{Config}\,\mathit{sig}\,C\,D\\ 
    &\mathsf{PAlgConfig}\,C\,D \defineeq (A : \Type) \times (B : A \to \Type) \times \mathsf{AlgConfig}\,A\,B\,C\,D
  \end{align*}
  for a configuration and algebra configuration over a signature. There are maps in both directions,
  \[
    \mathsf{SigConfig}\,C\,D \;\leftrightarrow\; \mathsf{PAlgConfig}\,C\,D,
  \]
  that do not form an equivalence because the map back always produces a W-type signature and the round-trip recovers the polynomial $(A, B)$ only up to equivalence.
\end{proposition}
\begin{proof}[Proof sketch]
  Both maps are the carrier-level bridge of Awodey et al.~\cite{awodey_homotopy-initial_2017} applied on each carrier, which identifies a W-type inductive algebra with a homotopy-initial algebra.
  Restricting to a W-type signature upgrades the maps to the type equivalence of~\Cref{thm: wtype-algequiv}.
\end{proof}

In our example and other repair case studies, the repaired program is always the \textit{same} program, with the universal property instantiated at the target and the rest of the data untouched.
The algebraic framework lets us formulate this as a general statement.

\begin{proposition}[Repair by Instantiation \agda{inductive-repair.palg-config}{foldSwap}]\label{prop: fold-swap}
  Given a polynomial $(A, B)$, fix an algebra configuration between $C$ and $D$ with repair map $\varphi$, and $(Y, \gamma) : \mathsf{P\text{-}Alg}\,A\,B$. Writing $\mathsf{fold}_C\,\gamma : C \to Y$ and $\mathsf{fold}_D\,\gamma : D \to Y$ for the unique homomorphisms out of each initial algebra,
  \[ \mathsf{fold}_C\,\gamma \;=\; \mathsf{fold}_D\,\gamma \circ \varphi. \]
\end{proposition}
\begin{proof}
  Both sides are homomorphisms $(C, \sigma_C) \to (Y, \gamma)$, and the hom-space out of an initial algebra is contractible.
\end{proof}

A program written against the universal property consists of defining data $\gamma$ together with an instantiation of $\mathsf{fold}$ at one initial algebra.
By \Cref{prop: fold-swap}, repair is re-instantiation: replace $\mathsf{fold}_C$ by $\mathsf{fold}_D$, keeping $\gamma$.
The repaired term $\mathsf{fold}_D\,\gamma$ is built from $D$'s universal property and $\gamma$ alone, so it mentions the source type $C$ exactly when $\gamma$ does.
When the program is written generically, its defining data is carrier-free (like the recursor cases \texttt{true} and $\lambda\_\to\mathtt{false}$), and the repaired term then mentions $C$ nowhere.
When $\gamma$ does mention $C$, those occurrences are further uses of $C$'s algebra structure, repairable by another instance of the same swap.
Whether a term mentions $C$ is a property of its syntax, which the type theory cannot express internally.

\Cref{prop: fold-swap} folds into a fixed algebra $(Y, \gamma)$, covering recursion but not induction, whose motive $P : C \to \Type$ mentions the carrier.
It applies anyway, since induction is a fold into the total-space algebra on $(x : C) \times P\,x$, with the eliminator its first-projection section \agda{inductive-repair.bridge}{bridgeBwd}.
Because the motive mentions $C$, the target is transported along the repair map (unique by \Cref{thm: wtype-algequiv}) rather than shared.
Projecting the resulting equation of pairs onto the proof component is exactly $\mathsf{elim\_ok}$ (\Cref{prop: correctness}), which is why that condition carries transported cases $\varphi^\ast\mathit{cases}$.
The transported motive still mentions $C$, but only through constructors and folds, themselves repairable by further swaps; iterating over motive and cases yields the native $D$-side proof, as in \texttt{zUnitB} of \Cref{sec: config}.

We view \Cref{prop: fold-swap} as stating that repair exploits a kind of \textit{representation independence}~\cite{angiuli_representation_2021} for algebras in type theory.
Programs and proofs written over a specific datatype are independent of the specific implementation of that datatype and are only sensitive to the underlying algebra, which is the key idea behind repair.
This is why repair tools can operate by syntactic instantiation, supplying the programs' defining data to the same universal property at a different carrier.
Thus, the algebraic perspective additionally gives a mathematical account of why repair can eliminate the source type as a dependency.

\subsection{Discussion: Repair without Univalence}\label{subsec: non-ua}

While our results have made real use of the homotopy interpretation of the intensional identity type, nowhere have we relied upon univalence: the entire inductive development is carried out in intensional Martin-L\"of type theory with function extensionality as the only axiom (\Cref{sec: conventions}).
This is notable because prior proof repair work was largely framed within a univalent context, where the correctness conditions were thought to require univalence~\cite{ringer2021pldi}.
Our development thus broadens the class of type theories in which this kind of proof repair is known to be justified. 
For example, our results apply to settings where univalence is inconsistent, such as Lean, whose type theory has uniqueness of identity proofs~\cite{lean_propositional_2026}.

\section{Dualizing Repair: Coinductive Types and Final Coalgebras}
\label{sec: coinductive}
With the algebraic perspective on proof repair in hand, we are able to dualize our development to coinductive types and final coalgebras.
Mechanized developments model infinite data and processes via coinductive types, which face the same burden of representation change that motivated repair for inductive types. 
Yet no existing repair tool addresses them.
In this section, we work within Cubical Agda~\cite{vezzosi_cubical_agda}, which supports a notion of equality that is amenable to coinductive types and bisimulation.

Coinductive types are defined in terms of destructors, so programs interact with coinductive data principally by observation.
Moreover, coinductive types enable reasoning about infinite objects through the principle of coinduction~\cite{KOZEN_SILVA_2017}.
A classic infinite object is a stream, an infinite sequence over a given type:
\[ (a_0, a_1, a_2, \ldots) : \mathsf{Stream}\, A \qquad A : \Type \qquad a_i : A .\]
The coinductive type of streams is defined as follows:
\begin{lstlisting}
  record Stream (A : Type) : Type where
    coinductive
    field
      hd : A
      tl : Stream A
\end{lstlisting}

We define programs that construct streams by specifying the observations that can be made of the stream---its head and its tail---a style known as copattern matching~\cite{abel_copattern}.
Here is the program that adjoins an element to the beginning of a stream:
\begin{lstlisting}
  cons : A -> Stream A -> Stream A
  hd (cons a s) = a
  tl (cons a s) = s
\end{lstlisting}

Using the principle of coinduction, we may verify that any stream can be decomposed into \texttt{cons} applied to its head and tail.
In Cubical Agda even this proof is built by copattern matching, a point we return to at the end of \Cref{subsec: warmup-coalg}:
\begin{lstlisting}
  streamComp : (s : Stream A) -> cons (hd s) (tl s) = s
  hd (streamComp s i) = hd s
  tl (streamComp s i) = tl s
\end{lstlisting}

Just like inductive types, coinductive types that capture the same data structure can have different definitions. 
For example, a library may also define the type of streams with a single destructor that packages the head and tail, along with its own \texttt{cons'}:

\noindent
\begin{minipage}[t]{0.52\textwidth}
\vspace{0pt}
\begin{lstlisting}
  record Stream' (A : Type) : Type where
    coinductive
    field
      obs : A $\times$ Stream' A
\end{lstlisting}
\end{minipage}
\hfill
\begin{minipage}[t]{0.44\textwidth}
\vspace{0pt}
\begin{lstlisting}
  cons' : A -> Stream' A -> Stream' A
  obs (cons' a s) = (a , s)
\end{lstlisting}
\end{minipage}

Our goal in this section is to lay the foundations for  proof repair algorithms that can handle changes between coinductive types, such as the change from \texttt{Stream} to \texttt{Stream'}.
We will work backward, utilizing the coalgebraic perspective on coinductive types to develop repair in terms of universal properties, and then extract a notion of configurations for coinductive types that exposes syntactic detail for an implementation.
\subsection{Warming Up to Coalgebras}\label{subsec: warmup-coalg}

We begin by dualizing the algebraic setup in \Cref{sec: alg-repair} to coalgebras.
We internalize the definitions given in standard sources~\cite{aczel_final_1989, rutten_universal_2000,Jacobs_Rutten_2011} on P-coalgebras within the type theory of Cubical Agda.

\begin{definition}[Coalgebra \agda{coinductive-repair.coalg}{P-Coalg}]
  Fix $A : \Type$ and $B : A \to \Type$.
  For the polynomial endofunctor $P\,A\,B$, a \textit{coalgebra} is a pair $(X, \gamma)$ of a carrier $X$ and a structure map $\gamma : X \to P\,A\,B\,X$.
  \begin{align*}
    &\mathsf{P\text{-}Coalg}\,A\,B : \Type\\
    &\mathsf{P\text{-}Coalg}\,A\,B \defineeq (X : \Type) \times (X \to P\,A\,B\,X)
  \end{align*}
\end{definition}

\begin{example}[Stream Coalgebra \agdamod{coinductive-repair.examples.stream}]
  Fix $A : \Type$ and consider the functor
  \[ P : \Type \to \Type \qquad P\,X \defineeq A \times X. \]
  A coalgebra for this functor is a pair $(Z, \gamma)$ where $Z : \Type$ and $\gamma : Z \to A \times Z$.
  We can understand this as a collection of streams over $A$ with dynamics given by $\gamma$, which, by the universal property of products, amounts to a pair \[\langle \gamma_1 , \gamma_2 \rangle : (Z \to A) \times (Z \to Z).\]
  Thus, we have a ``stream system'' given by a collection of streams, an output function $\gamma_1$ (which gives the head of a stream), and a transition function $\gamma_2$ (which gives the tail of a stream).
\end{example}

\begin{remark}
  Fix $A : \Type$ and  $B : A \to \Type$.
  Because the definition of polynomial endofunctor (\Cref{def: poly-endo}) is the dependent sum given by $A$ and $B$, we can always present a coalgebra structure map $\gamma : Z \to P\, A\, B\, Z$ as a
  \textit{shape} $\gamma_1 : Z \to A$ and a \textit{position} $\gamma_2 : (z : Z) \to B\,(\gamma_1\, z) \to Z$.
\end{remark}

\begin{definition}[Coalgebra Homomorphism \agda{coinductive-repair.coalg}{isCoalgHom}]
  Given a polynomial $(A, B)$ and two coalgebras $(X, \gamma)$ and $(Y, \delta)$ for $P\,A\,B$, a \textit{coalgebra homomorphism} from $(X, \gamma)$ to $(Y, \delta)$ is a function $f : X \to Y$ such that the following diagram commutes:
  \[
    \begin{tikzcd}
      X \arrow{r}{\gamma} \arrow{d}{f} & P X \arrow{d}{P f} \\
      Y \arrow{r}{\delta} & P Y
    \end{tikzcd}
  \]
  Formally, we define the type of coalgebra homomorphisms as the function space between the carriers paired with the proof of the above commutative diagram:
  \begin{align*}
    &\mathsf{isCoalgHom} : (X\,Y : \mathsf{P\text{-}Coalg}\,A\,B) \to (\mathsf{Carrier}\, X \to \mathsf{Carrier}\, Y) \to \Type\\
    &\mathsf{isCoalgHom}\, (X , \gamma)\, (Y , \delta)\, f \defineeq \mathsf{P\text{-}map}\, f \circ \gamma = \delta \circ f \\[4pt]
    &\mathsf{CoalgHom} : \mathsf{P\text{-}Coalg}\,A\,B \to \mathsf{P\text{-}Coalg}\,A\,B \to \Type\\
    &\mathsf{CoalgHom}\,X \, Y \defineeq (f : \mathsf{Carrier}\, X \to \mathsf{Carrier}\, Y) \times \mathsf{isCoalgHom}\, X\, Y\, f
  \end{align*}
\end{definition}

\begin{example}[Stream Homomorphism \agdamod{coinductive-repair.examples.stream}]
  Homomorphisms between coalgebras are maps that preserve and reflect the dynamics of the source coalgebra. 
  This means a homomorphism of stream coalgebras preserves outputs (heads) and transitions (tails):
  \[\begin{tikzcd}
    Z && Y \\
    {A \times Z} && {A \times Y}
    \arrow["{h }"', from=1-1, to=1-3]
    \arrow["{\langle \gamma_1, \gamma_2 \rangle }", from=1-1, to=2-1]
    \arrow["{\langle \sigma_1, \sigma_2 \rangle}"', from=1-3, to=2-3]
    \arrow["{1_A \times h}", from=2-1, to=2-3]
  \end{tikzcd}\]
\end{example}

Coalgebras support a notion of finality dual to the initiality of algebras.
We introduce two notions that will become useful for our development: a {final} coalgebra, whose universal property pins down only the underlying function of the unique homomorphism, and a {homotopy-final} coalgebra, which additionally pins down the homomorphism witness.

\begin{definition}[Final Coalgebra \agda{coinductive-repair.coalg}{isFinal}]
  Given a polynomial $(A, B)$, a coalgebra $X : \mathsf{P\text{-}Coalg}\,A\,B$ is \textit{final} if for every $Y : \mathsf{P\text{-}Coalg}\,A\,B$ there exists a coalgebra homomorphism $Y \to X$, and the {underlying function} of that homomorphism is unique.
  \begin{align*}
    &\mathsf{isFinal} : \mathsf{P\text{-}Coalg}\,A\,B \to \Type\\
    &\mathsf{isFinal}\, X \defineeq (Y : \mathsf{P\text{-}Coalg}\, A\, B) \to (h : \mathsf{CoalgHom}\, Y\, X) \\
    &\qquad\qquad\qquad \times \big((h' : \mathsf{CoalgHom}\, Y \, X) \to \fst{h} = \fst{h'}\big)
  \end{align*}
\end{definition}

\begin{definition}[Homotopy-Final Coalgebra \agda{coinductive-repair.coalg}{isHFinal}]
  Given a polynomial $(A, B)$, a coalgebra $X : \mathsf{P\text{-}Coalg}\,A\,B$ is \textit{homotopy-final} if for every $Y : \mathsf{P\text{-}Coalg}\,A\,B$ the type of coalgebra homomorphisms $Y \to X$ is contractible.
  \begin{align*}
    &\mathsf{isHFinal} : \mathsf{P\text{-}Coalg}\, A\, B \to \Type\\
    &\mathsf{isHFinal}\, X \defineeq (Y : \mathsf{P\text{-}Coalg}\, A\, B) \to \mathsf{isContr}(\mathsf{CoalgHom}\, Y\, X)
  \end{align*}
\end{definition}

The reason for distinguishing between these two notions of finality is that proof assistant support for coinductive types that have homotopy-finality is not fully developed.
Cubical Agda formalizations have confirmed finality~\cite{damato_formalising_2025} for containers, but whether coinductive records in Agda have homotopy-finality is open. 
As far as we know, whether coinductive types in Rocq have finality is open, and Lean would trivially have homotopy-finality (given finality) since it supports the uniqueness of identity proofs axiom.
State-of-the-art results~\cite{ahrens_nonwellfounded_2015, gylterud_terminal_2026} formulate M-types internally to type theory, but not in a fashion that is suitable for repair of native definitions.
We return to this point when developing repair for coinductive types.

\begin{example}[Final Stream Coalgebra \agdamod{coinductive-repair.examples.stream}]\label{ex: stream-final}
  The \texttt{Stream} type is the final coalgebra for the stream functor $P\,S \defineeq X \times S$.
  Its coalgebra structure is just the two destructors bundled up, unfolding a stream into its head and tail:
  \begin{lstlisting}
    outStream : Stream X -> P X $\unit$ (Stream X)
    outStream s = (hd s , $\lambda$ _ -> tl s)
  \end{lstlisting}

  The universal property of final coalgebras is what enables the construction of coprograms. 
  The homomorphism from a coalgebra to a final coalgebra gives coinductive types a corecursion principle, a method for defining functions from any coalgebra to a final coalgebra by specifying how the function interacts with the destructors of the coalgebra.
  For \texttt{Stream}, the head is read off from the structure map of the source coalgebra $C$, and the tail corecurses on the next state it produces:
  \begin{lstlisting}
    Stream-corec : (C : P-Coalg X $\unit$) -> CoalgCarrier C -> Stream X
    hd (Stream-corec C c) = fst (CoalgOut C c)
    tl (Stream-corec C c) = Stream-corec C (snd (CoalgOut C c) tt)
  \end{lstlisting}
  That \texttt{Stream-corec}\,$C$ is a homomorphism holds definitionally and it is the unique such function.
\end{example}

The verification of finality for \texttt{Stream} requires reasoning about the equality of two streams to prove the uniqueness of a coalgebra homomorphism.
Equality for the infinite systems given by coalgebra (or coinductive types) requires a notion of indistinguishable behavior that is not captured by the usual intensional identity type: that type is inductively generated by $\mathsf{refl}$, so its only closed proofs identify terms that are already definitionally equal, and it offers no way to {build} an identification of two infinite objects from an observation-by-observation comparison.
To understand equality on coalgebras, we define an adequate behavior-preserving relation on them.

\begin{definition}[Bisimulation \agda{coinductive-repair.coalg}{Bisim}]\label{def: bisim}
  Given a polynomial $(A, B)$, fix a coalgebra $(X, \langle\gamma_1, \gamma_2\rangle)$ for $P\,A\,B$ and a relation $R : X \to X \to \Type$. The terms $x, y : X$ satisfy the one-step condition $\mathsf{Bisim}\,R\,x\,y$ when they have equal shape observations and $R$-related position observations.
  \begin{align*}
    &\mathsf{Bisim}\,R : (x\,y : X) \to \Type\\
    &\mathsf{Bisim}\,R\,x\,y \defineeq (\text{b-shape} : \gamma_1\,x = \gamma_1\,y) \times \big((b : B\,(\gamma_1\,x)) \to R\,(\gamma_2\,x\,b)\,(\gamma_2\,y\,b^{*})\big)
  \end{align*} where $b^{*}$ indicates a transport across b-shape.

  The relation $R$ is a \textit{bisimulation} when every $R$-related pair satisfies this one-step condition.
  \begin{align*}
    &\mathsf{isBisim} : (R : X \to X \to \Type) \to \Type\\
    &\mathsf{isBisim}\,R \defineeq \forall \{x\,y\} \to R\,x\,y \to \mathsf{Bisim}\,R\,x\,y
  \end{align*}
  Terms $x, y : X$ are then {bisimilar} when they are related by some bisimulation, some $R$ with $\mathsf{isBisim}\,R$ and $R\,x\,y$.
  That this notion of bisimulation is a bisimulation in the categorical sense of Rutten~\cite{rutten_universal_2000} is proven in the subsequent proposition.
\end{definition}

The definition of bisimulation given here is the single-carrier specialization of the definition typically given~\cite{rutten_universal_2000}. 
We use this specialization for simplicity; our use of the definition always transports relations over one carrier to another, so we have no need for a heterogeneous relation. 
\begin{proposition}[Coinduction Principle for Final Coalgebras\agda{coinductive-repair.coalg}{coind-from-isFinal}]\label{prop: coind}
  Fix a final coalgebra $(X, \gamma)$.
  The coinduction principle says that bisimilar elements are equal. Given a relation $R$ that is a bisimulation, every $R$-related pair is identified by a path.
  \begin{align*}
    &\mathsf{coind} : (R : X \to X \to \Type) \to \mathsf{isBisim}\,R \\
    &\qquad\qquad \to \forall \{x\,y\} \to R\,x\,y \to x = y
  \end{align*}
\end{proposition}
\begin{proof}[Proof sketch]
  This is the argument of Rutten~\cite{rutten_universal_2000}, repeated in type theory by Ahrens et al.~\cite{ahrens_nonwellfounded_2015}.
  The argument goes by forming the coalgebra on  $\widehat{R} \defineeq \Sigma\,x\,y.\,R\,x\,y$, whose projections are coalgebra homomorphisms into $(X, \gamma)$, which must coincide; evaluating at $x, y$ gives the identity.
\end{proof}

The proof of \Cref{prop: coind} relies on $\mathsf{isBisim}\,R$ giving a coalgebra structure on the total space $(\widehat{R}, \gamma_R)$ whose projections $\mathsf{fst}, \mathsf{fst} \circ \mathsf{snd} : \widehat{R} \to X$ are coalgebra homomorphisms into $(X, \gamma)$.
Our lifting of the relation $R$ is the canonical one for a polynomial endofunctor $P\,A\,B$, in the sense that $\mathsf{Bisim}$ is the weakest one-step condition from which this follows: $\mathsf{Bisim}$ yields such a $\gamma_R$, and any other condition that does implies $\mathsf{Bisim}$~\agda{coinductive-repair.coalg}{Bisim-weakest}.

\begin{example}[Bisimulation on Streams \agdamod{coinductive-repair.examples.stream}]
  \label{ex: stream-bisim}
  For streams the bisimulation data specializes to a relation $R$ relating streams whose heads agree and whose tails are $R$-related:
  \begin{lstlisting}
    record StreamBis (R : Stream X -> Stream X -> Type) (s t : Stream X) where
      field
        bhd : hd s = hd t
        btl : R (tl s) (tl t)
  \end{lstlisting}
  Coinduction then turns any such bisimulation into a path, defined by copattern matching on the path's head and tail:
  \begin{lstlisting}
    Stream-coind : (R : Stream X -> Stream X -> Type)
                 -> ($\forall$ {s t} -> R s t -> StreamBis R s t)
                 -> $\forall$ {s t} -> R s t -> s = t
    hd (Stream-coind R isBis r i) = bhd (isBis r) i
    tl (Stream-coind R isBis r i) = Stream-coind R isBis (btl (isBis r)) i
  \end{lstlisting}
\end{example}

The ability to construct proofs of identity via copattern matching as in \texttt{streamComp} and \Cref{ex: stream-bisim} is the principal reason for our choice of Cubical Agda.
The path type of cubical type theory is not an inductively generated type; it is defined as a function space out of the interval, which facilitates introduction by means other than $\mathsf{refl}$.
On the other hand, the usual identity type is amenable to the sort of breaking-down of structure that pattern matching on inductive types provides.
Indeed, the bisimulation and the path type are equivalent in Cubical Agda (coind is an equivalence)~\cite{vezzosi_cubical_agda}, giving a type theory where equality is suited for coalgebraic reasoning.

\subsection{Coalgebraic Repair}

Finally, we have the tools in hand to cast coinductive repair in terms of final coalgebras.
The goal of repair is to establish a relationship between types across which programs and proofs can be repaired.
Coalgebra gives us the definitive notion of a relationship between coalgebras: our goal is to define a homomorphism between the coalgebras corresponding to \texttt{Stream} and \texttt{Stream'} so that we can transfer the programs and proofs between the two.
However, as in the inductive case, the proof engineer would like an efficient repair mechanism that removes the source type as a dependency.
We develop this in two parts: first with a coalgebra configuration, and then with an algorithm-friendly configuration.

\begin{definition}[(Homotopy) Coalgebra Configuration \agda{coinductive-repair.coconfig-coalg-equiv}{FinalCoalgOn}]
  Given a polynomial $(A, B)$, a {final coalgebra structure} on a carrier $C$ packages a destructor map with a proof that the resulting coalgebra is final:
  \[ \mathsf{FinalCoalgOn}\,A\,B\,C \defineeq (\gamma : C \to P\,A\,B\,C) \times \mathsf{isFinal}\,(C, \gamma). \]
  A {coalgebra configuration} between carriers $C$ and $D$ is a pair of final coalgebra structures for the same polynomial, one on each carrier:
  \[ \mathsf{CoalgConfig}\,A\,B\,C\,D \defineeq \mathsf{FinalCoalgOn}\,A\,B\,C \;\times\; \mathsf{FinalCoalgOn}\,A\,B\,D. \]
  A coalgebra configuration determines a coalgebra isomorphism: each side's finality yields a unique homomorphism into the other, and the two round-trips are forced to be the identities by uniqueness.
  A homotopy coalgebra configuration is given by replacing the finality condition with $\mathsf{isHFinal}$.
\end{definition}
Now, we can use such a configuration to do the swap induced by the underlying equivalence on programs and proofs, rather than compose with the equivalence.
\begin{example}[Coalgebraic Repair with Stream and Stream' \agdamod{coinductive-repair.examples.stream}]\label{ex: stream-repair}
  The configuration equips each carrier with the same universal property, so a program written with the corecursor and a proof written with the coinduction principle are repaired by swapping the universal property at which they are instantiated.
  The defining data of \texttt{cons} mentions its carrier only through the structure map, so it can be stated once, against an abstract coalgebra $(E, \mathsf{out})$:
  \begin{lstlisting}
    consSeed : P-Coalg X $\unit$ -> P-Coalg X $\unit$
    consSeed (E , out) = ((X $\times$ E) $\uplus$ E , $\gamma$)
      where $\gamma$ (inl (y , t)) = (y , $\lambda$ _ -> inr t)
            $\gamma$ (inr t)       = (fst (out t) , $\lambda$ _ -> inr (snd (out t) tt))

    cons-M  : X -> Stream  X -> Stream  X
    cons-M  x s = Stream-corec  (consSeed (Stream  X , outStream )) (inl (x , s))

    cons'-M : X -> Stream' X -> Stream' X
    cons'-M x s = Stream'-corec (consSeed (Stream' X , outStream')) (inl (x , s))

    streamComp-M  : (s : Stream  X) -> cons-M  (hd s) (tl s) = s
    streamComp-M  s = Stream-coind  etaR  etaStep  (refl , replay-M  (tl s))

    streamComp'-M : (s : Stream' X) -> cons'-M (fst (obs s)) (snd (obs s)) = s
    streamComp'-M s = Stream'-coind etaR' etaStep' (refl , replay'-M (snd (obs s)))
  \end{lstlisting}
  The coalgebra \texttt{consSeed} starts in state \texttt{inl (x , s)}, emits \texttt{x}, and then moves to \texttt{inr s}, from which it replays \texttt{s} through \texttt{out}.
  The tail of \texttt{cons-M x s} is the replay of \texttt{s}, which is identified with \texttt{s} by $\mathtt{replay\text{-}M}/\mathtt{replay'\text{-}M}$.
  This mirrors the observation after \Cref{prop: fold-swap}: occurrences of the source type in defining data are further uses of its (co)algebra structure, repaired by the same swap.
  The bisimulations $\mathtt{etaR}/\mathtt{etaR'}$ likewise read their carriers only through the structure maps.
  The repaired program and proof carry no reference to the source type.
\end{example}

\subsection{Coinductive Repair via Configurations}
With the coalgebraic perspective on coinductive types, we can now develop the framework for repair algorithms for coinductive types. Akin to the standard presentation of W-types, we give coinductive types a set of type-theoretic rules, and refer to types satisfying these rules as M-types.
The state-of-the-art work on M-types (in the formulation of rules) has not established their homotopy-finality.
In fact, the contemporary work has focused on proving the finality and homotopy-finality for various formulations of M-types~\cite{ahrens_nonwellfounded_2015,damato_formalising_2025,gylterud_terminal_2026}, none of which are rule-based presentations.
We are primarily interested in developing repair for coinductive types, and thus are concerned with rule-based presentations that facilitate programming and proving.
To develop repair for contemporary and future type theories, our development handles both a notion of M-types that corresponds to final coalgebras and a stronger notion that corresponds to homotopy-final coalgebras.

An M-type for a polynomial $(A, B)$ is a carrier, written $\mathsf{M}\,A\,B$, equipped with destructors
\[ \mathsf{shape} : \mathsf{M}\,A\,B \to A \qquad\qquad \mathsf{pos} : (x : \mathsf{M}\,A\,B) \to B\,(\mathsf{shape}\,x) \to \mathsf{M}\,A\,B \]
satisfying the rules $\mathsf{corec}$, $\beta_s$, $\beta_p$, and $\eta$ of \Cref{fig: mrules}.
The corecursor $\mathsf{corec}$ maps any coalgebra $(C, \langle s, p \rangle)$ into $\mathsf{M}\,A\,B$, and its two computation ($\beta$) rules say that destructing a corecursive value computes: $\beta_s$ for the shape and $\beta_p$ for the positions, which lives over $\beta_s$ and is hence a dependent path.
The uniqueness ($\eta$) rule states that any map agreeing with the source coalgebra under the destructors, that is, any coalgebra homomorphism, equals the corecursor.
Strengthening $\eta$ from the function level to the $\Sigma$ level gives $\eta_\Sigma$, which we use for homotopy M-types (\Cref{def: hmtypes}).
Every M-type also enjoys a coinduction principle $\mathsf{coind}$, which says that any bisimulation entails a path; unlike the other rules, it is not assumed but derived (\Cref{def: mtypes}).

\begin{figure}
  \centering
  \begin{mathpar}
    \inferrule*[right=$\mathsf{corec}$]{s : C \to A \\ p : (c : C) \to B\,(s\,c) \to C}
              {\mathsf{corec}\,s\,p : C \to \mathsf{M}\,A\,B}
    \and
    \inferrule*[right=$\beta_s$]{ }{\beta_{s} : \mathsf{shape}\,(\mathsf{corec}\,s\,p\,c) = s\,c}
    \and
    \inferrule*[right=$\beta_p$]{ }
              {\beta_{p} : \mathsf{pos}\,(\mathsf{corec}\,s\,p\,c) \;=_{\beta_s}\; \mathsf{corec}\,s\,p \circ (p\,c)}
    \and
    \inferrule*[right=$\eta$]
      {h : C \to \mathsf{M}\,A\,B \\\\ \beta^h_{s} : (c : C) \to \mathsf{shape}\,(h\,c) = s\,c \\ \beta^h_{p} : (c : C) \to \mathsf{pos}\,(h\,c) =_{\beta^h_s\,c} h \circ (p\,c)}
      {\eta : h = \mathsf{corec}\,s\,p}
    \and
    \inferrule*[right=$\eta_\Sigma$]
      {(\eta\ \text{premises})}
      {(h, \beta^h_s, \beta^h_p) \;=\; (\mathsf{corec}\,s\,p,\ \beta_s,\ \beta_p)}
    \and
    \inferrule*[right=$\mathsf{coind}$]
      {R : \mathsf{M}\,A\,B \to \mathsf{M}\,A\,B \to \Type \\ b : \mathsf{isBisim}\, R\, \langle \mathsf{shape}, \mathsf{pos} \rangle \\ r : R\,x\,y}
      {\mathsf{coind}\,R\,b\,r : x = y}
  \end{mathpar}
  \caption{Rules of the (homotopy) M-type structure over a polynomial $(A, B)$: the corecursor $\mathsf{corec}$, its computation rules $\beta_s$ and $\beta_p$, the function-level uniqueness rule $\eta$ and its $\Sigma$-level strengthening $\eta_\Sigma$, and the derived coinduction principle $\mathsf{coind}$.}
  \label{fig: mrules}
\end{figure}

\newif\ifmrulesrecord
\mrulesrecordfalse
\begin{definition}[M-Types \agda{coinductive-repair.coalg}{HasMRules}]\label{def: mtypes}
 The type $\mathsf{HasMRules}$ bundles the rules of M-types~(\Cref{fig: mrules}):
\ifmrulesrecord
  \begin{align*}
    &\mathbf{record}\ \mathsf{HasMRules}\,X\,\mathsf{shape}\,\mathsf{pos} : \Type\ \mathbf{where}\\
    &\quad \mathsf{corec} : \{C : \Type\}\,(s : C \to A)\,(p : (c : C) \to B\,(s\,c) \to C) \to C \to X\\
    &\quad \beta_s : \forall \{C\,s\,p\,c\} \to \mathsf{shape}\,(\mathsf{corec}\,s\,p\,c) = s\,c\\
    &\quad \beta_p : \forall \{C\,s\,p\,c\} \to \mathsf{pos}\,(\mathsf{corec}\,s\,p\,c) =_{\beta_s} \mathsf{corec}\,s\,p \circ (p\,c)\\
    &\quad \eta : \forall \{C\,s\,p\}\,(h : C \to X) \to (\beta^h_s : (c : C) \to \mathsf{shape}\,(h\,c) = s\,c)\\
    &\quad\qquad \to ((c : C) \to \mathsf{pos}\,(h\,c) =_{\beta^h_s\,c} h \circ (p\,c)) \to h = \mathsf{corec}\,s\,p
  \end{align*}
\else
  \begin{align*}
    &\mathsf{HasMRules} : (X : \Type) \to (\mathsf{shape} : X \to A) \to ((x : X) \to B\,(\mathsf{shape}\,x) \to X) \to \Type\\
    &\mathsf{HasMRules}\,X\,\mathsf{shape}\,\mathsf{pos} \defineeq\\
    &\qquad\phantom{\times}\ (\mathsf{corec} : \{C : \Type\}\,(s : C \to A)\,(p : (c : C) \to B\,(s\,c) \to C) \to C \to X)\\
    &\qquad \times (\beta_s : \forall \{C\,s\,p\,c\} \to \mathsf{shape}\,(\mathsf{corec}\,s\,p\,c) = s\,c)\\
    &\qquad \times (\beta_p : \forall \{C\,s\,p\,c\} \to \mathsf{pos}\,(\mathsf{corec}\,s\,p\,c) =_{\beta_s} \mathsf{corec}\,s\,p \circ (p\,c))\\
    &\qquad \times \big(\eta : \forall \{C\,s\,p\}\,(h : C \to X) \to (\beta^h_s : (c : C) \to \mathsf{shape}\,(h\,c) = s\,c)\\
    &\qquad\qquad \to ((c : C) \to \mathsf{pos}\,(h\,c) =_{\beta^h_s\,c} h \circ (p\,c)) \to h = \mathsf{corec}\,s\,p\big)
  \end{align*}
\fi
  An \textit{M-type} is a carrier with destructors for which $\mathsf{HasMRules}$ is inhabited. 
  As we will see (\Cref{thm: mtypes-final}), M-types are carriers of final coalgebras, and, as such, enjoy a coinduction principle, so we do not include a coinduction principle as a field of $\mathsf{HasMRules}$.
\end{definition}

We emphasize that M-types are not a primitive of our ambient theory: $\mathsf{HasMRules}$ is a {property} that a given carrier with given destructors may or may not enjoy.
There exists an M-type for any polynomial $(A, B)$, and any M-type has a coalgebra given by $\langle \mathsf{shape}, \mathsf{pos} \rangle$, which enables the formation of bisimulations on M-types, as in the derived $\mathsf{coind}$ principle.

Concretely, we use the coinductive records of Cubical Agda to model these M-types and verify that they satisfy the given rules \agda{coinductive-repair.coalg}{M-HasMRules}.
\begin{lstlisting}
  record M (A : Type) (B : A -> Type) : Type where
    coinductive
    field
      shape : A
      pos   : B shape -> M A B
\end{lstlisting}
\begin{lstlisting}
  corec : (C : Type) (s : C -> A) (p : (c : C) -> B (s c) -> C) -> C -> M A B
  shape (corec C s p c)   = s c
  pos   (corec C s p c) b = corec C s p (p c b)

  M-HasMRules : HasMRules (M A B) shape pos
  M-HasMRules = isFinal$\to$HasMRules M-isFinal

  M-coind = HasMRules-coind M-HasMRules
\end{lstlisting}
Now, we develop M-types that directly correspond to homotopy-final coalgebras.
\begin{definition}[Homotopy M-Types \agda{coinductive-repair.coalg}{HasHMRules}]\label{def: hmtypes}
  A {homotopy M-type} strengthens the $\eta$ rule of an M-type from the function level to the $\Sigma$ level: a homomorphism $(h, \beta_h)$ agreeing with the source coalgebra under the destructors equals the corecursor bundle $(\mathsf{corec}\,s\,p, \beta)$ by a $\Sigma$-path ($\eta_\Sigma$ in \Cref{fig: mrules}).
  Exactly as $\mathsf{HasMRules}$ bundles the plain rules, we bundle the M-rules with $\eta_\Sigma$ in place of $\eta$ into a single type $\mathsf{HasHMRules}\,X\,\mathsf{shape}\,\mathsf{pos}$.
\end{definition}

A homotopy M-type differs from an M-type only in this one strengthened rule.
Note that while the existence of homotopy M-types in HoTT is known via the limit construction of Ahrens et al.~\cite{ahrens_nonwellfounded_2015}, our development and prior work do not prove that they exist within Cubical Agda.
Proving that the native $\texttt{M}$ record is a homotopy M-type requires characterizing the path types of $\texttt{M}$ at higher dimensions, which the state of the art~\cite{damato_formalising_2025,gylterud_terminal_2026} does not do; indeed, our attempted formalization ran into a documented issue~\cite{altenkirch_cubicalissue_2020}.
The obstruction disappears under additional hypotheses on the polynomial. For example, $\texttt{Stream}\,A$ is a set whenever $A$ is a set \agdamod{coinductive-repair.examples.stream}. The same fact holds for M-types whose shape type is a set, but we do not pursue the general statement here.

Configurations expose the syntactic detail of a particular representation that a repair algorithm walks over. 
We give a configuration framework for coinductive types by focusing on destructors and outputs. 
We introduce the destructor-level signature and algebra, then define a single carrier's coinductive structure, and finally give a configuration as a pair of these.

\begin{definition}[Destructor Algebra \agda{coinductive-repair.config}{DestrAlgebra}]
  The output telescope of a single destructor is described by a \textit{destructor arity} $\mathsf{DestrArity}$, generated by the same constructors $\mathsf{Done}$, $\mathsf{Nonrec}$, and $\mathsf{Rec}$ as $\mathsf{ConstrArity}$ (\Cref{def: constr-alg}).
  Reading an arity off at a carrier $X$ gives the \textit{outputs} $\mathsf{Outputs}\,\mathit{cs}\,X$ of such a destructor, defined clause for clause as $\mathsf{Args}$, with functorial action $\mathsf{mapOutputs}$ as for $\mathsf{mapArgs}$.
  A {CoSignature} $\mathit{sig}$ is a type $\mathsf{Op}$ of destructor labels together with an arity assignment $\mathsf{arity} : \mathsf{Op} \to \mathsf{DestrArity}$.
  Where a constructor algebra consumes arguments, a \textit{destructor algebra} for $\mathit{sig}$ on a carrier $X$ interprets every label as a destructor exposing its outputs:
  \[ \mathsf{DestrAlgebra}\,\mathit{sig}\,X \defineeq \textstyle (op : \mathsf{Op}) \to \big(X \to \mathsf{Outputs}\,(\mathsf{arity}\,op)\,X\big). \]
\end{definition}

\begin{definition}[Coinductive Configuration \agda{coinductive-repair.config}{CoConfig}]\label{def: coind-config}
   A \textit{coinductive coalgebra} carries a destructor algebra equipped with corecursion and computation rules. 
   A $\mathsf{CoindCoalg}$ for $\mathit{sig}$ on $X$ packages:
  \begin{itemize}
    \item $\mathsf{destr} : \mathsf{DestrAlgebra}\,\mathit{sig}\,X$, the syntactic coalgebra on $X$;
    \item $\mathsf{corec} : \mathsf{DestrAlgebra}\,\mathit{sig}\,C \to C \to X$ for any source carrier $C$, the corecursor;
    \item $\mathsf{corec\text{-}\beta}$, the per-destructor computation rule: for every source coalgebra $\delta$ and label $op$,
      \[ \mathsf{destr}\,op\,(\mathsf{corec}\,\delta\,c) \;=\; \mathsf{mapOutputs}\,(\mathsf{arity}\,op)\,(\mathsf{corec}\,\delta)\,(\delta\,op\,c); \]
    \item $\mathsf{corec\text{-}\eta}$, function-level uniqueness: any $h$ whose destructors agree with $\delta$ equals $\mathsf{corec}\,\delta$.
  \end{itemize}
  Again, in Cubical Agda the coinduction principle $\mathsf{coind}$ is derivable from these fields, so $\mathsf{CoindCoalg}$ already carries it implicitly. 
  A \textit{coinductive configuration} for $\mathit{sig}$ between carriers $C$ and $D$ is a pair of coinductive coalgebras, one on each carrier:
  \[ \mathsf{CoConfig}\,\mathit{sig}\,C\,D \defineeq \mathsf{CoindCoalg}\,\mathit{sig}\,C \;\times\; \mathsf{CoindCoalg}\,\mathit{sig}\,D. \]
\end{definition}

Coinductive configurations admit a homotopy strengthening just like the M-type formulations.

\begin{definition}[Homotopy Coinductive Configuration \agda{coinductive-repair.config}{HCoConfig}]
  An {$\mathsf{HCoindCoalg}$} replaces the function-level $\mathsf{corec\text{-}\eta}$ with the {$\Sigma$-level} uniqueness rule $\eta_\Sigma$ of a homotopy M-type ($\mathsf{corec\text{-}uniq\text{-}}\Sigma$).
  A {homotopy coinductive configuration} is a pair of these:
  \[ \mathsf{HCoConfig}\,\mathit{sig}\,C\,D \defineeq \mathsf{HCoindCoalg}\,\mathit{sig}\,C \;\times\; \mathsf{HCoindCoalg}\,\mathit{sig}\,D. \]
\end{definition}

As with inductive configurations, we can verify that the coinductive configurations satisfy the correctness conditions a proof engineer would expect of a repair algorithm. 
These conditions are new correctness criteria for coinductive repair, but they follow the idea that the configuration behaves coherently with the underlying equivalence.

\begin{proposition}[Correctness Conditions \agda{coinductive-repair.coherence}{CoCoherence}]\label{prop: co-correctness}
  Fix a coinductive configuration $\mathsf{CoConfig}\,\mathit{sig}\,C\,D$ with destructor coalgebras $\mathsf{destr}_C, \mathsf{destr}_D$, corecursors $\mathsf{corec}_C, \mathsf{corec}_D$, and $\beta$ laws.
  Let $\varphi : C \to D$ be the forward map of the equivalence induced by the configuration. It satisfies the coinductive duals of the PUMPKIN~Pi coherence conditions.
  \begin{enumerate}
    \item $\mathsf{destr\_ok}$. The destructors commute with $\varphi$: for every label $op$ and $c : C$,
      \[ \mathsf{destr}_D\,op\,(\varphi\,c) \;=\; \mathsf{mapOutputs}\,(\mathsf{arity}\,op)\,\varphi\,(\mathsf{destr}_C\,op\,c). \]
      Since $\varphi = \mathsf{corec}_D\,\mathsf{destr}_C$ by construction, this is exactly its per-destructor $\beta$ rule.
    \item $\mathsf{iota\_ok}$. Each corecursor computes on its own destructors: this is exactly $\mathsf{corec\text{-}\beta}$, one law for $\mathsf{corec}_C$ and one for $\mathsf{corec}_D$.
    \item $\mathsf{corec\_ok}$. The corecursors commute with $\varphi$: for every source coalgebra $\mathit{ec} : \mathsf{DestrAlgebra}\,\mathit{sig}\,E$ and $e : E$,
      \[ \varphi\,(\mathsf{corec}_C\,\mathit{ec}\,e) \;=\; \mathsf{corec}_D\,\mathit{ec}\,e. \]
    \item $\mathsf{coind\_ok}$. The coinduction principles commute with $\varphi$: writing $\mathsf{coind}_C, \mathsf{coind}_D$ for coinduction on each side, for every bisimulation $R$ on $C$ and $r : R\,x\,y$,
      \[ \mathsf{ap}\,\varphi\,(\mathsf{coind}_C\,r) \;=\; \mathsf{coind}_D\,(\varphi_\ast r), \]
      an equality of the parallel paths $\varphi\,x = \varphi\,y$, where $\varphi_\ast$ transports the $C$-bisimulation to a $D$-bisimulation.
      It only holds for a homotopy configuration ($\mathsf{HCoConfig}$).
  \end{enumerate}
\end{proposition}
\begin{proof}[Proof sketch]
  The first three follow directly from the structure the configuration carries; the interesting part is $\mathsf{coind\_ok}$.
    It lifts the argument of \Cref{prop: coind} one dimension: a two-dimensional bisimulation $R_2$ on parallel paths $p, q : a = b$ yields $p = q$~\agda{coinductive-repair.config}{deriveCoind2D}.
    The coalgebra now lives on $\Sigma\,a\,b\,(p\,q : a = b).\,R_2\,p\,q$, and in place of two projections there is, for each $k : \mathbb{I}$, a pair of homomorphisms sending $(a, b, p, q, r)$ to $p\,k$ and $q\,k$.
    The $\Sigma$-level uniqueness of corecursors ($\mathsf{corec\text{-}uniq\text{-}\Sigma}$) identifies all of these with the corecursor, coherently at $k = 0, 1$, which yields the square $p = q$ and is why $\mathsf{coind\_ok}$ needs the homotopy tier.
    Applying this on $D$ to $\mathsf{ap}\,\varphi\,(\mathsf{coind}_C\,r)$ and $\mathsf{coind}_D\,(\varphi_\ast r)$, related by a lift of $R$, gives $\mathsf{coind\_ok}$.
\end{proof}

Beyond generalizing Rutten's argument, the two-dimensional principle holds for any carrier satisfying the interface $\mathsf{HasHMRules}$, not just a particular construction of M-types such as the limit construction of Ahrens et al.~\cite{ahrens_nonwellfounded_2015}, and it makes no truncation assumption.
Working at the level of the interface is crucial for a theory of repair, which must apply to the coinductive types a proof assistant provides natively rather than to one chosen construction.
The fourth condition, $\mathsf{coind\_ok}$, is the one condition that demands the homotopy tier.
Streams over a set assemble into an $\mathsf{HCoConfig}$ \agdamod{coinductive-repair.examples.stream}, so $\mathsf{coind\_ok}$ holds in particular for the running \texttt{Stream}/\texttt{Stream'} example whenever the element type is a set.

The $\mathsf{corec\_ok}$ condition is the coinductive form of repair by instantiation (\Cref{prop: fold-swap}), the dual of $\mathsf{fold}_C\,\gamma = \mathsf{fold}_D\,\gamma \circ \varphi$.
It can take the source coalgebra $\mathit{ec}$ verbatim because $\mathit{ec}$ lives on a third type $E$ and never mentions the repaired carriers, just as a fold's target algebra $\gamma$ does not.
By contrast, the data of $\mathsf{elim\_ok}$ (a motive over $C$) and of $\mathsf{coind\_ok}$ (a bisimulation on $C$) mention the carrier, and so must be transported along $\varphi$.

\begin{example}[Stream and Stream' Repair via Configurations \agdamod{coinductive-repair.examples.stream}]
  Now we can phrase our running example through coinductive configurations.
  \begin{lstlisting}
    StreamSig : CoSignature
    StreamSig = MTypeSignature X $\unit$

    streamCoindCoalg  : CoindCoalg StreamSig (Stream  X)
    streamCoindCoalg  = coalgBridgeBwd (outStream  , Stream-isFinal)
    stream'CoindCoalg : CoindCoalg StreamSig (Stream' X)
    stream'CoindCoalg = coalgBridgeBwd (outStream' , Stream'-isFinal)

    streamCoConfig : CoConfig StreamSig (Stream X) (Stream' X)
    streamCoConfig = record { coindCoalgC = streamCoindCoalg; coindCoalgD = stream'CoindCoalg }
  \end{lstlisting}
  Given any source coalgebra, \texttt{unfold ind} corecurses into that component's carrier. 
  Written against the abstract $\mathsf{CoindCoalg}$, it names neither \texttt{Stream} nor \texttt{Stream'}.
  \begin{lstlisting}
    unfold : CoindCoalg StreamSig E -> DestrAlgebra StreamSig S -> S -> E
    unfold ind = CoindCoalg.corec ind
  \end{lstlisting}
  Repair is now executed through a swap of the configuration's components.
  \texttt{consVia} writes \texttt{cons} against an abstract component.
  Instantiating at \texttt{streamCoindCoalg} yields \texttt{cons} on \texttt{Stream}; the repair is the \emph{same program} instantiated at the swapped component, now landing in \texttt{Stream'}.
  \begin{lstlisting}
    consVia : CoindCoalg StreamSig E -> X -> E -> E
    consVia ind x e = unfold ind $\delta$ (inl (x , e))
      where $\delta$ op (inl (y , t)) = y , ($\lambda$ _ -> inr t) , tt
            $\delta$ op (inr t)       = mapOutputs _ inr (destr ind op t)

    consCfg  : X -> Stream  X -> Stream  X
    consCfg  = consVia streamCoindCoalg
    consCfg' : X -> Stream' X -> Stream' X
    consCfg' = consVia stream'CoindCoalg  
  \end{lstlisting}
  The proof \texttt{streamComp} is repaired the same way, through the coinduction principle each component carries (\Cref{def: coind-config}).
  As in \Cref{ex: stream-repair}, the seed now carries the replay path, \texttt{replay-id}.
  \begin{lstlisting}
    streamCompCfg  : (s : Stream  X) -> consCfg  (hd s) (tl s) = s
    streamCompCfg  s = deriveCoind streamCoindCoalg  etaR  etaRIsBisim
                         (refl , replay-id streamCoindCoalg  (tl s))

    streamCompCfg' : (s : Stream' X) -> consCfg' (fst (obs s)) (snd (obs s)) = s
    streamCompCfg' s = deriveCoind stream'CoindCoalg etaR' etaRIsBisim'
                         (refl , replay-id stream'CoindCoalg (snd (obs s)))
  \end{lstlisting}
  \looseness=-1
  This models the execution of a proof repair algorithm for coinductive types on the running programs and proofs.
  These repaired terms compute as we would expect.
  Observing one reduces the configuration scaffolding away, leaving native \texttt{Stream'} data.
  Each equation below holds definitionally.
  \begin{lstlisting}
    consCfg'-head : (x : X) (s : Stream' X) -> fst (obs (consCfg' x s)) = x
    consCfg'-next : (x : X) (s : Stream' X)$\hspace{-.08cm}$ -> $\hspace{-.07cm}$fst (obs (snd (obs (consCfg' x s)))) = fst (obs s)
  \end{lstlisting}
\end{example}

\subsection{Connecting Coinductive Configurations and Coalgebra}

We now characterize M-types as carriers of final coalgebras, setting the foundation for cementing our slogan that proof repair is an algebra isomorphism in the case of coinductive repair.
We give the coalgebraic dual of Awodey et al.'s~\cite{awodey_homotopy-initial_2017} characterization of W-types as homotopy-initial algebras.
The novelty is the left-hand side of the equivalence below, a rule-based presentation of (homotopy) M-types, whereas prior treatments establish finality for coalgebraic or limit constructions of M-types~\cite{ahrens_nonwellfounded_2015, damato_formalising_2025, gylterud_terminal_2026}, not for such an interface.

\begin{theorem}[(Homotopy) M-Types are (Homotopy) Final Coalgebras]\label{thm: mtypes-final}
  Given a polynomial $(A, B)$ and a carrier $X$ with destructors $(s : X \to A,\ p : (x : X) \to B\,(s\,x) \to X)$, having the (homotopy) M-rules is exactly the induced coalgebra $(X, \langle s, p \rangle)$ being (homotopy-)final.
  \begin{enumerate}
    \item (Final \agda{coinductive-repair.coalg}{HasMRules-\%E2\%89\%83-isFinal})
      $\mathsf{HasMRules}\,X\,s\,p \;\simeq\; \mathsf{isFinal}\,(X, \lambda x.\,(s\,x, p\,x))$
    \item (Homotopy-final \agdamod{coinductive-repair.coalg}) $\mathsf{HasHMRules}\,X\,s\,p \;\simeq\; \mathsf{isHFinal}\,(X, \lambda x.\,(s\,x, p\,x))$
  \end{enumerate}
\end{theorem}
\begin{proof}[Proof sketch]
  The corecursor together with its two $\beta$ rules is exactly a coalgebra homomorphism out of any source, and $\eta$ is exactly function-level uniqueness of that homomorphism, and conversely $\mathsf{corec}$, $\beta_s$, $\beta_p$, and $\eta$ are read back off the finality data. 
  Both round-trips hold definitionally.
  The homotopy-final tier strengthens $\eta$ from function-level uniqueness to contractibility of the coalgebra-homomorphism space. Contractibility transports either way and both sides are propositions, so the bi-implication is an equivalence.
\end{proof}

On the inductive side, the correspondence between W-types and initial algebras requires technical higher-dimensional path algebra because the W-type data has a dependent eliminator that quantifies over all motives. 
For the coinductive story, the correspondence given by \Cref{thm: mtypes-final} is much easier to establish because of the absence of a dependent co-eliminator; instead, the difficulty is relocated to deriving that extra data in the proof of $\mathsf{coind\_ok}$.

Coinductive configurations have a canonical form we call \textit{M-type configurations}. 
These configurations emulate M-types by capturing the coinductive structure of types with a single destructor that exposes a non-recursive and a recursive output.

\begin{definition}[M-type Configuration \agda{coinductive-repair.coconfig-coalg-equiv}{MTypeSignature}]
  We define a one-destructor {M-type signature}: a single observation exposing a shape $a : A$ and a $B\,a$-indexed family of recursive positions.
  \begin{align*}
    &\mathsf{MTypeArity} : (A : \Type) \to (A \to \Type) \to \unit \to \mathsf{DestrArity}\\
    &\mathsf{MTypeArity}\,A\,B\,\_ \defineeq \mathsf{Nonrec}\,A\,(\lambda a \to \mathsf{Rec}\,(B\,a)\,\mathsf{Done})\\[4pt]
    &\mathsf{MTypeSignature} : (A : \Type) \to (A \to \Type) \to \mathsf{CoSignature}\\
    &\mathsf{MTypeSignature}\,A\,B\, \mathsf{.CoSignature.Op} \defineeq \unit\\
    &\mathsf{MTypeSignature}\,A\,B\, \mathsf{.CoSignature.arity} \defineeq \mathsf{MTypeArity}\,A\,B\\[4pt]
    &\mathsf{MTypeCoindCoalg} : (A : \Type) \to (A \to \Type) \to (C : \Type) \to \Type\\
    &\mathsf{MTypeCoindCoalg}\,A\,B\,C \defineeq \mathsf{CoindCoalg}\,(\mathsf{MTypeSignature}\,A\,B)\,C
  \end{align*}
  An {M-type configuration} is the special coinductive configuration whose signature is an M-type signature, collapsing all per-destructor data into one destructor with one non-recursive output and one recursive output:
  \begin{align*}
    &\mathsf{MTypeCoConfig} : (A : \Type) \to (A \to \Type) \to (C\,D : \Type) \to \Type\\
    &\mathsf{MTypeCoConfig}\,A\,B\,C\,D \defineeq \mathsf{CoConfig}\,(\mathsf{MTypeSignature}\,A\,B)\,C\,D
  \end{align*}
\end{definition}

M-type configurations collapse the detail of a configuration into the canonical one-destructor form over which we identify configurations with coalgebra isomorphisms.
Every coinductive configuration is faithfully represented by an M-type configuration.

\begin{lemma}[Coinductive Configurations Faithfully Embed into M-type Configurations \agdamod{coinductive-repair.mtype-config}]\label{lem: embed-mtype}
  Every coinductive configuration over an arbitrary coinductive signature $\mathit{sig}$ determines an M-type configuration on the same carriers.
  We read off the polynomial from the signature: 
  \[ A \;\defineeq\; ({op : \mathsf{Op}}) \to \mathsf{Outputs}(\mathsf{arity}\,op)\,\unit, \qquad B\,a \;\defineeq\; ({op : \mathsf{Op}}) \times \mathsf{recPos}(\mathsf{arity}\,op)\,(a\,op), \]
  and the configuration converts to an M-type configuration for this $(A, B)$:
  \[ \mathsf{CoConfig}\,\mathit{sig}\,C\,D \to ({A : \Type}) \times (({B : (A \to \Type)}) \times \mathsf{MTypeCoConfig}\,A\,B\,C\,D). \]
  The embedding is {faithful}: the forward and backward maps of the repair equivalence induced by the M-type configuration are equal to those induced by the original configuration.
\end{lemma}

We now state our main theorem over M-type configurations, using the coinductive dual of \Cref{def: alg-equiv} in two tiers matching the two notions of finality.
A \emph{coalgebra equivalence} $\mathsf{CoalgEquiv}\,X\,Y$ \agdamod{coinductive-repair.coconfig-coalg-equiv} is a type equivalence of the carriers whose forward and backward maps both carry coalgebra-homomorphism witnesses.
Its homotopy counterpart, $\mathsf{CoalgEquiv^H}$ \agdamod{coinductive-repair.coconfig-coalg-equiv}, dualizes \Cref{def: alg-equiv}.

\begin{theorem}[Coinductive Configurations are Coalgebra Isomorphisms \agdamod{coinductive-repair.coconfig-coalg-equiv}]
  Over the M-type signature, a coinductive configuration is a coalgebra isomorphism whose source is final.
  \begin{enumerate}
    \item (Final) A configuration is logically equivalent to an isomorphism of final coalgebras:
      \begin{align*}
        \mathsf{CoConfig}\,(\mathsf{MTypeSignature}\,A\,B)\,C\,D \;\longleftrightarrow\;
        & (\gamma_C : C \to P\,A\,B\,C) \times (\gamma_D : D \to P\,A\,B\,D)\\
        &\times \mathsf{CoalgEquiv}\,(C, \gamma_C)\,(D, \gamma_D) \times \mathsf{isFinal}\,(C, \gamma_C)
      \end{align*}
    \item (Homotopy-final) A {homotopy} configuration is a coalgebra isomorphism between homotopy-final coalgebras:
      \begin{align*}  
        \mathsf{HCoConfig}\,(\mathsf{MTypeSignature}\,A\,B)\,C\,D \;\simeq\;
        & (\gamma_C : C \to P\,A\,B\,C) \times (\gamma_D : D \to P\,A\,B\,D)\\
        &\times \mathsf{CoalgEquiv^H}\,(C, \gamma_C)\,(D, \gamma_D) \times \mathsf{isHFinal}\,(C, \gamma_C)
      \end{align*}
  \end{enumerate}
\end{theorem}
\begin{proof}[Proof sketch]
  The first link is a carrier-level bridge given by our characterization of (homotopy) M-types as (homotopy-)final coalgebras (\Cref{thm: mtypes-final}).
  We then identify a pair of final coalgebras with one final coalgebra plus a coalgebra isomorphism to the other. 
  The two directions are mutually inverse exactly when the matched data is propositional. 
  This holds at the homotopy-final tier, giving an equivalence. 
  It fails at the final tier, where $\mathsf{isFinal}$ does not pin the homomorphism witness.
\end{proof}

\looseness=-1
This theorem is the coinductive completion of our slogan; a coinductive configuration witnesses that its carriers are final for a common polynomial functor.
Repair between coinductive types is therefore natural precisely when the types have natural presentations as final coalgebras for the same functor, and the correspondence used for repair is the coalgebra isomorphism that finality forces to exist uniquely between them.
At the homotopy-final tier a coalgebra isomorphism out of a final source is a proposition (as the proof above observes), so a homotopy coinductive configuration, too, determines a unique repair map.
The changes handled by coinductive configurations are exactly those representable by a coalgebra isomorphism, the changes of destructor presentation that preserve the observable behavior of the type.

This development largely rests upon the discovery that proof repair for inductive types corresponds to algebra isomorphisms.
Before the connection between repair and P-algebras was established, it was unclear how to extend proof repair to coinductive types, since the question was muddled by the many different implementation approaches taken by various proof assistants, none of which isolates what a coinductive type is.
By focusing on the universal properties that coinductive types have, we are able to construct a high-level foundation for proof repair algorithms in a general manner, independent of any specific implementation approach.

\section{Related Work}
\label{sec: related}

\hspace{\parindent}\looseness=-1\textbf{Proof Repair.} Our work gives the first foundational treatment of proof repair.
Proof repair was first introduced in 2018~\cite{ringer2018cpp, robert2018}.
It is the proof engineering~\cite{qedatlarge} analogue of program repair~\cite{Monperrus2015, Anwar2026}.
Some of the early examples of proof repair~\cite{ringer2019orn} were based on the theory of algebraic ornaments~\cite{ornaments2011, dagand2012, dagand2013, ko2016}, which express a relation between indexed datatypes where one index type has richer information than the other.
The class of proof repair algorithms that our foundations support comes from a line of work on automatically
repairing proofs in Rocq in response to breaking changes in datatypes~\cite{ringer2021pldi, viola2025oopsla, ringer_proof_2021, ringer2019orn}. That line of work focuses on algorithmic approaches, but leaves open questions about what a configuration represents and what makes a configuration ``natural.''
The thesis work in particular~\cite{ringer_proof_2021} presents a diagram that gives intuition that inspired our foundations, but goes no further in developing that intuition.
We formalize this line of work and answer these open questions. 
Supplementary automation beyond the core repair functionality, like decompilation, is out of the scope of our foundations.

Other proof repair algorithms exist that are out of the scope of our foundations.
For example, \textsc{Sisyphus}~\cite{gopinathan_mostly_2023} repairs separation logic proofs about imperative programs written in OCaml. 
Proof Mate~\cite{proofmate} repairs proofs written in PVS.
Preliminary work explores proof repair for JML~\cite{arr-repair}.
All of these tools concern specification and proof languages that are quite different from the dependent type theory that we concern ourselves with.

\looseness=-1
Recently, neural methods for proof repair have begun to emerge, both as standalone tools~\cite{reichel2023, patel2026proofrepairbench, tahat2024proof, yang2026exverus, ospanov2026apollo, wang2026learningrepairleanproofs, reichel2024neural} and as a way of augmenting other models for proofs~\cite{First2023, r2026lampleanbasedagenticframework, jana2026proofbridge}. The scope of these methods is defined by their training data, rather than by any foundational system. In the long run, foundations like ours may help in defining kinds of data that these tools might wish to support, or in building well-scoped and general neurosymbolic tools that integrate symbolic repair into other neural models (in the fashion of PALM~\cite{palm}).

\textbf{Proof Reuse and Transfer.} Proof repair is one way of reusing proofs. The key conditions that yield repair rather than the more general reuse are that (1) proof repair is primarily concerned with adapting proofs to changes rather than with reusing proofs across developments, and (2) proof repair is most useful in the presence of imperfect foresight or breaking changes in dependencies that are outside of one's control. Giving up the first condition 
yields tools for proof reuse, like refinement frameworks~\cite{refinementsforfree} and transfer methods~\cite{isabelletransfer, marriage, equivalencesforfree, cohen2024trocq}; giving up the second condition yields design principles that aim to make proofs reusable or robust to changes to begin with~\cite{cpdt, Delaware2011, Woos2016, qedatlarge}. 

\looseness=-1
The relationships between these tools and technologies have been discussed extensively in prose~\cite{viola2025oopsla, ringer_proof_2021, qedatlarge, cohen2024trocq}. Our foundations 
add a formal dimension to this discussion. By formalizing the connection between equivalences, configurations, and algebras, we draw an explicit bridge between proof repair and transfer methods whose inner workings use automatic transport. Moreover, with \Cref{prop: fold-swap}, we learn that proof repair
satisfies its distinguishing conditions by exploiting a kind of representation independence~\cite{angiuli_representation_2021}. This property may also hold of other ``white-box'' approaches to automatic transfer~\cite{marriage}, which can be used for repair.

\textbf{(Co)algebras and (Co)inductive Types.}
We briefly cover previous work on (co)algebras and (co)inductive types from which our development greatly benefited, with no claim of completeness.

\looseness=-1
Inductive types are a well-researched area within programming languages and type theory. The seminal work of Dybjer~\cite{dybjer_inductive} gives a schema for inductive and recursive definitions within Martin-L\"of type theory (MLTT).
The work of Chapman et al.~\cite{chapman_levitation} presents a dependent type theory which gives inductive definitions via an encoding in a universe.
The connection between inductive datatypes, algebras, and polynomial endofunctors has received significant attention. 
The primary inspiration for our work is the discovery by Awodey et al.~\cite{awodey_homotopy-initial_2017} that the rules for W-types are equivalent to the existence of homotopy-initial algebras for polynomial endofunctors.
Abbott et al.~\cite{ABBOTT2005} show that the notion of containers (similar to polynomial functors) can be used to represent all strictly positive types and the polymorphic functions between them.
Damato et al.~\cite{damato_formalising_2025} generalize the container development to a setting without the uniqueness of identity proofs axiom, giving a formulation in intensional type theory and formalization in Cubical Agda.

\looseness=-1
Coinductive types and coalgebras are a prominent and growing research area.
The work of Hagino~\cite{hagino} uses categorical logic to develop a simply typed lambda calculus with algebraic and coalgebraic datatypes.
Basold and Geuvers~\cite{basold} extend Hagino's categorical datatypes to a dependent type theory constructed solely from inductive and coinductive types, while leaving the (co)induction principles out.
Ahrens et al.~\cite{ahrens_nonwellfounded_2015} establish that the existence of M-types (in the formulation of homotopy-final coalgebras) is derivable in HoTT in the presence of inductive types. 
In the course of developing models of non-wellfounded sets in HoTT, Gylterud et al.~\cite{gylterud_terminal_2026} characterize the identity type of the M-type (formulated as a homotopy-final coalgebra) as an indexed M-type.
Most modern proof assistants feature some form of coinductive types, made possible by various research programs~\cite{gimenez, abel_copattern,laarakker_constructing_2026}.

\section{Conclusion}
\label{sec: conclusion}

We have given the first foundational treatment of type-theoretic proof repair.
Its fundamental construct, the configuration, is exactly an algebra isomorphism between initial algebras for a common polynomial functor (\Cref{thm: wtype-algequiv}). 
This discovery greatly clarifies the foundational understanding of repair: the coherence conditions imposed by prior tools are consequences of initiality rather than side obligations, the repair equivalence a configuration determines is unique, and programs written against the universal property are repaired by re-instantiation alone (\Cref{prop: fold-swap}).
Dualizing this account yielded the foundations of repair for a class of types no existing tool handles---coinductive types---together with new correctness criteria for coinductive repair and a novel characterization of (homotopy) M-types as (homotopy-)final coalgebras (\Cref{thm: mtypes-final}).

Two aspects of the development deserve emphasis.
First, none of it uses univalence: the inductive side lives in intensional MLTT with function extensionality, and the coinductive side uses Cubical Agda only for its path type, not for univalent principles.
Since prior repair work was framed univalently, this broadens the class of type theories in which repair is known to be justified.
Second, the algebraic framing is extensible.
Higher inductive types, quotient inductive types, inductive-inductive types, and inductive-recursive types all admit algebraic presentations over richer signatures, and we expect configurations, their coherence conditions, and the uniqueness of repair to transfer along the same lines.
We hope this foundation serves both as an explanation of why existing repair tools work and as a blueprint for the repair tools that do not yet exist.

\section*{Acknowledgments}
This material is based upon work supported by the Defense Advanced Research Projects Agency (DARPA)
under Contract No. HR0011262E019. 
Any opinions, findings, and conclusions or recommendations expressed in this material are those of the authors and do not necessarily reflect the views of DARPA.
This work is also supported by the UK Advanced Research + Invention Agency (ARIA).

\section*{Disclosure of Use of AI in Research}
We used AI tools (specifically, Claude) in our research for two tasks: 
(1) the standard refactoring of Agda proofs, and 
(2) making sense of the implications of \Cref{prop: fold-swap}.
For the second task, after we had most of the main results, we asked Claude
whether any of them could explain why repair can eliminate old types.
This is a property of repair that had captured our curiosity,
but for which we had thus far failed to find a crisp, formal explanation.
Claude responded by drawing a connection between our result of \Cref{prop: fold-swap} and 
this property.

\section*{Data Availability Statement}
The artifact for our paper is available on Zenodo~\cite{plante2026artifact}.
It includes the Agda formalization of all of our results, as well as the Cubical Agda examples
from \Cref{sec: coinductive}.
An explanation of the artifact and how to install it can be found in \texttt{artifact/README.md}.
The explanation also includes a centralized table relating our informal and formal results
(as an alternative presentation of the blue checkmarks throughout the paper that link to the formal results directly).
Our data is additionally available in a public GitHub repository.\footnote{\url{https://github.com/trey3p/proof-repair-coalgebra}}

\bibliographystyle{ACM-Reference-Format}
\bibliography{references}


\begin{thebibliography}{62}

%%% except for \shownote{} and \showURL{}.  The latter does
%%% not use final punctuation, in order to avoid confusing it with
%%% the Web address.

\ifx \showCODEN    \undefined \def \showCODEN     #1{\unskip}     \fi
\ifx \showISBNx    \undefined \def \showISBNx     #1{\unskip}     \fi
\ifx \showISBNxiii \undefined \def \showISBNxiii  #1{\unskip}     \fi
\ifx \showISSN     \undefined \def \showISSN      #1{\unskip}     \fi
\ifx \showLCCN     \undefined \def \showLCCN      #1{\unskip}     \fi
\ifx \shownote     \undefined \def \shownote      #1{#1}          \fi
\ifx \showarticletitle \undefined \def \showarticletitle #1{#1}   \fi
\ifx \showURL      \undefined \def \showURL       {\relax}        \fi
% invisible to TeX
\providecommand\bibfield[2]{#2}
\providecommand\bibinfo[2]{#2}
\providecommand\natexlab[1]{#1}
\providecommand\showeprint[2][]{arXiv:#2}
\makeatletter
\@ifundefined{NAT@parse@date}{}{\let\NAT@parse@date@orig\NAT@parse@date}
\@ifundefined{NAT@parse@date}{}{\def\NAT@parse@date#1#2#3#4#5#6@@{\NAT@parse@date@orig#1#2#3#4#5#6@@\def\NAT@tempyear{0000}\def\NAT@tempexlab{{?}}\ifx\NAT@year\NAT@tempyear\ifx\NAT@exlab\NAT@tempexlab\def\NAT@date{[n.\,d.]}\else\edef\NAT@date{[n.\,d.]\NAT@exlab}\fi\fi}}
\makeatother

\bibitem[alt(2020)]%
        {altenkirch_cubicalissue_2020}
 \bibinfo{year}{2020}\natexlab{}.
\newblock \bibinfo{booktitle}{\emph{cubical primitives should preserve
  guardedness · Issue \#4740 · agda/agda}}.
\newblock
\urldef\tempurl%
\url{https://github.com/agda/agda/issues/4740}
\showURL{%
\tempurl}


\bibitem[lea(2026)]%
        {lean_propositional_2026}
 \bibinfo{year}{2026}\natexlab{}.
\newblock \bibinfo{booktitle}{\emph{Propositional Equality}}.
\newblock
\urldef\tempurl%
\url{https://lean-lang.org/doc/reference/latest/Basic-Propositions/Propositional-Equality/#UIP}
\showURL{%
\tempurl}


\bibitem[roc(2026)]%
        {rocq}
 \bibinfo{year}{2026}\natexlab{}.
\newblock \bibinfo{booktitle}{\emph{The Rocq Prover 9.2 documentation}}.
\newblock
\urldef\tempurl%
\url{https://rocq-prover.org/doc/V9.2.0/refman/index.html}
\showURL{%
\tempurl}


\bibitem[Abbott et~al\mbox{.}(2005)]%
        {ABBOTT2005}
\bibfield{author}{\bibinfo{person}{Michael Abbott}, \bibinfo{person}{Thorsten
  Altenkirch}, {and} \bibinfo{person}{Neil Ghani}.}
  \bibinfo{year}{2005}\natexlab{}.
\newblock \showarticletitle{Containers: Constructing strictly positive types}.
\newblock \bibinfo{journal}{\emph{Theoretical Computer Science}}
  \bibinfo{volume}{342}, \bibinfo{number}{1} (\bibinfo{year}{2005}),
  \bibinfo{pages}{3--27}.
\newblock
\showISSN{0304-3975}
\href{https://doi.org/10.1016/j.tcs.2005.06.002}{doi:\nolinkurl{10.1016/j.tcs.2005.06.002}}
\newblock
\shownote{Applied Semantics: Selected Topics}.


\bibitem[Abel et~al\mbox{.}(2013)]%
        {abel_copattern}
\bibfield{author}{\bibinfo{person}{Andreas Abel}, \bibinfo{person}{Brigitte
  Pientka}, \bibinfo{person}{David Thibodeau}, {and} \bibinfo{person}{Anton
  Setzer}.} \bibinfo{year}{2013}\natexlab{}.
\newblock \showarticletitle{Copatterns: programming infinite structures by
  observations}. In \bibinfo{booktitle}{\emph{Proceedings of the 40th Annual
  ACM SIGPLAN-SIGACT Symposium on Principles of Programming Languages}} (Rome,
  Italy) \emph{(\bibinfo{series}{POPL '13})}. \bibinfo{publisher}{Association
  for Computing Machinery}, \bibinfo{address}{New York, NY, USA},
  \bibinfo{pages}{27–38}.
\newblock
\showISBNx{9781450318327}
\href{https://doi.org/10.1145/2429069.2429075}{doi:\nolinkurl{10.1145/2429069.2429075}}


\bibitem[Aczel and Mendler(1989)]%
        {aczel_final_1989}
\bibfield{author}{\bibinfo{person}{Peter Aczel} {and} \bibinfo{person}{Nax
  Mendler}.} \bibinfo{year}{1989}\natexlab{}.
\newblock \showarticletitle{A final coalgebra theorem}. In
  \bibinfo{booktitle}{\emph{Category Theory and Computer Science}} (Berlin,
  Heidelberg, 1989), \bibfield{editor}{\bibinfo{person}{David~H. Pitt},
  \bibinfo{person}{David~E. Rydeheard}, \bibinfo{person}{Peter Dybjer},
  \bibinfo{person}{Andrew~M. Pitts}, {and} \bibinfo{person}{Axel Poigné}}
  (Eds.). \bibinfo{publisher}{Springer}, \bibinfo{pages}{357--365}.
\newblock
\showISBNx{978-3-540-46740-3}
\href{https://doi.org/10.1007/BFb0018361}{doi:\nolinkurl{10.1007/BFb0018361}}


\bibitem[Adámek et~al\mbox{.}(2025)]%
        {Adámek_Milius_Moss_2025}
\bibfield{author}{\bibinfo{person}{Jiří Adámek}, \bibinfo{person}{Stefan
  Milius}, {and} \bibinfo{person}{Lawrence~S. Moss}.}
  \bibinfo{year}{2025}\natexlab{}.
\newblock \bibinfo{booktitle}{\emph{Initial Algebras and Terminal Coalgebras:
  The Theory of Fixed Points of Functors}}.
\newblock \bibinfo{publisher}{Cambridge University Press}.
\newblock
\href{https://doi.org/10.1017/9781108884112}{doi:\nolinkurl{10.1017/9781108884112}}


\bibitem[Ahrens et~al\mbox{.}(2015)]%
        {ahrens_nonwellfounded_2015}
\bibfield{author}{\bibinfo{person}{Benedikt Ahrens}, \bibinfo{person}{Paolo
  Capriotti}, {and} \bibinfo{person}{R{\'e}gis Spadotti}.}
  \bibinfo{year}{2015}\natexlab{}.
\newblock \showarticletitle{Non-Wellfounded Trees in Homotopy Type Theory}. In
  \bibinfo{booktitle}{\emph{13th International Conference on Typed Lambda
  Calculi and Applications (TLCA 2015)}} \emph{(\bibinfo{series}{Leibniz
  International Proceedings in Informatics (LIPIcs)},
  Vol.~\bibinfo{volume}{38})}. \bibinfo{publisher}{Schloss
  Dagstuhl--Leibniz-Zentrum f{\"u}r Informatik}, \bibinfo{pages}{17--30}.
\newblock
\href{https://doi.org/10.4230/LIPIcs.TLCA.2015.17}{doi:\nolinkurl{10.4230/LIPIcs.TLCA.2015.17}}


\bibitem[Angiuli et~al\mbox{.}(2021)]%
        {angiuli_representation_2021}
\bibfield{author}{\bibinfo{person}{Carlo Angiuli}, \bibinfo{person}{Evan
  Cavallo}, \bibinfo{person}{Anders M\"ortberg}, {and} \bibinfo{person}{Max
  Zeuner}.} \bibinfo{year}{2021}\natexlab{}.
\newblock \showarticletitle{Internalizing representation independence with
  univalence}.
\newblock \bibinfo{journal}{\emph{Proceedings of the ACM on Programming
  Languages}} \bibinfo{volume}{5}, \bibinfo{number}{POPL}, Article
  \bibinfo{articleno}{12} (\bibinfo{year}{2021}).
\newblock
\href{https://doi.org/10.1145/3434293}{doi:\nolinkurl{10.1145/3434293}}


\bibitem[Anwar(2026)]%
        {Anwar2026}
\bibfield{author}{\bibinfo{person}{Ashif Anwar}.}
  \bibinfo{year}{2026}\natexlab{}.
\newblock \showarticletitle{Neuro Symbolic Automated Program Repair: A
  Systematic Review of LLM-Based and Symbolic Techniques}.
\newblock \bibinfo{journal}{\emph{International Journal of Computer Engineering
  and Data Science (IJCEDS)}} \bibinfo{volume}{5}, \bibinfo{number}{1}
  (\bibinfo{date}{Mar.} \bibinfo{year}{2026}), \bibinfo{pages}{1–16}.
\newblock
\urldef\tempurl%
\url{https://www.ijceds.com/ijceds/article/view/118}
\showURL{%
\tempurl}


\bibitem[Awodey(2010)]%
        {awodey_book}
\bibfield{author}{\bibinfo{person}{Steve Awodey}.}
  \bibinfo{year}{2010}\natexlab{}.
\newblock \bibinfo{booktitle}{\emph{Category Theory} (\bibinfo{edition}{2nd}
  ed.)}.
\newblock \bibinfo{publisher}{Oxford University Press, Inc.},
  \bibinfo{address}{USA}.
\newblock
\showISBNx{0199237182}
\href{https://doi.org/10.5555/2060081}{doi:\nolinkurl{10.5555/2060081}}


\bibitem[Awodey et~al\mbox{.}(2017)]%
        {awodey_homotopy-initial_2017}
\bibfield{author}{\bibinfo{person}{Steve Awodey}, \bibinfo{person}{Nicola
  Gambino}, {and} \bibinfo{person}{Kristina Sojakova}.}
  \bibinfo{year}{2017}\natexlab{}.
\newblock \showarticletitle{Homotopy-Initial Algebras in Type Theory}.
\newblock  \bibinfo{volume}{63}, \bibinfo{number}{6} (\bibinfo{year}{2017}),
  \bibinfo{pages}{1--45}.
\newblock
\showISSN{0004-5411, 1557-735X}
\href{https://doi.org/10.1145/3006383}{doi:\nolinkurl{10.1145/3006383}}


\bibitem[Basold and Geuvers(2016)]%
        {basold}
\bibfield{author}{\bibinfo{person}{Henning Basold} {and}
  \bibinfo{person}{Herman Geuvers}.} \bibinfo{year}{2016}\natexlab{}.
\newblock \showarticletitle{Type Theory based on Dependent Inductive and
  Coinductive Types}. In \bibinfo{booktitle}{\emph{Proceedings of the 31st
  Annual ACM/IEEE Symposium on Logic in Computer Science}} (New York, NY, USA)
  \emph{(\bibinfo{series}{LICS '16})}. \bibinfo{publisher}{Association for
  Computing Machinery}, \bibinfo{address}{New York, NY, USA},
  \bibinfo{pages}{327–336}.
\newblock
\showISBNx{9781450343916}
\href{https://doi.org/10.1145/2933575.2934514}{doi:\nolinkurl{10.1145/2933575.2934514}}


\bibitem[Chapman et~al\mbox{.}(2010)]%
        {chapman_levitation}
\bibfield{author}{\bibinfo{person}{James Chapman},
  \bibinfo{person}{Pierre-\'{E}variste Dagand}, \bibinfo{person}{Conor
  McBride}, {and} \bibinfo{person}{Peter Morris}.}
  \bibinfo{year}{2010}\natexlab{}.
\newblock \showarticletitle{The gentle art of levitation}. In
  \bibinfo{booktitle}{\emph{Proceedings of the 15th ACM SIGPLAN International
  Conference on Functional Programming}} (Baltimore, Maryland, USA)
  \emph{(\bibinfo{series}{ICFP '10})}. \bibinfo{publisher}{Association for
  Computing Machinery}, \bibinfo{address}{New York, NY, USA},
  \bibinfo{pages}{3–14}.
\newblock
\showISBNx{9781605587943}
\href{https://doi.org/10.1145/1863543.1863547}{doi:\nolinkurl{10.1145/1863543.1863547}}


\bibitem[Chlipala(2013)]%
        {cpdt}
\bibfield{author}{\bibinfo{person}{Adam Chlipala}.}
  \bibinfo{year}{2013}\natexlab{}.
\newblock \bibinfo{booktitle}{\emph{Certified Programming with Dependent Types
  - {A} Pragmatic Introduction to the Coq Proof Assistant}}.
\newblock \bibinfo{publisher}{{MIT} Press}.
\newblock
\showISBNx{978-0-262-02665-9}
\href{https://doi.org/10.7551/mitpress/9153.001.0001}{doi:\nolinkurl{10.7551/mitpress/9153.001.0001}}


\bibitem[Cohen et~al\mbox{.}(2024)]%
        {cohen2024trocq}
\bibfield{author}{\bibinfo{person}{Cyril Cohen}, \bibinfo{person}{Enzo Crance},
  {and} \bibinfo{person}{Assia Mahboubi}.} \bibinfo{year}{2024}\natexlab{}.
\newblock \showarticletitle{Trocq: proof transfer for free, with or without
  univalence}. In \bibinfo{booktitle}{\emph{European Symposium on
  Programming}}. Springer, \bibinfo{pages}{239--268}.
\newblock
\href{https://doi.org/10.1007/978-3-031-57262-3_10}{doi:\nolinkurl{10.1007/978-3-031-57262-3_10}}


\bibitem[Cohen et~al\mbox{.}(2013)]%
        {refinementsforfree}
\bibfield{author}{\bibinfo{person}{Cyril Cohen}, \bibinfo{person}{Maxime
  D{\'e}n{\`e}s}, {and} \bibinfo{person}{Anders M{\"o}rtberg}.}
  \bibinfo{year}{2013}\natexlab{}.
\newblock \showarticletitle{Refinements for Free!}. In
  \bibinfo{booktitle}{\emph{Certified Programs and Proofs}},
  \bibfield{editor}{\bibinfo{person}{Georges Gonthier} {and}
  \bibinfo{person}{Michael Norrish}} (Eds.). \bibinfo{publisher}{Springer
  International Publishing}, \bibinfo{address}{Cham},
  \bibinfo{pages}{147--162}.
\newblock
\showISBNx{978-3-319-03545-1}
\href{https://doi.org/10.1007/978-3-319-03545-1_10}{doi:\nolinkurl{10.1007/978-3-319-03545-1_10}}


\bibitem[Dagand and McBride(2012)]%
        {dagand2012}
\bibfield{author}{\bibinfo{person}{Pierre-Evariste Dagand} {and}
  \bibinfo{person}{Conor McBride}.} \bibinfo{year}{2012}\natexlab{}.
\newblock \showarticletitle{Transporting functions across ornaments}. In
  \bibinfo{booktitle}{\emph{Proceedings of the 17th ACM SIGPLAN International
  Conference on Functional Programming}} (Copenhagen, Denmark)
  \emph{(\bibinfo{series}{ICFP '12})}. \bibinfo{publisher}{Association for
  Computing Machinery}, \bibinfo{address}{New York, NY, USA},
  \bibinfo{pages}{103–114}.
\newblock
\showISBNx{9781450310543}
\href{https://doi.org/10.1145/2364527.2364544}{doi:\nolinkurl{10.1145/2364527.2364544}}


\bibitem[Dagand and McBride(2013)]%
        {dagand2013}
\bibfield{author}{\bibinfo{person}{Pierre-Evariste Dagand} {and}
  \bibinfo{person}{Conor McBride}.} \bibinfo{year}{2013}\natexlab{}.
\newblock \showarticletitle{A Categorical Treatment of Ornaments}. In
  \bibinfo{booktitle}{\emph{Proceedings of the 2013 28th Annual ACM/IEEE
  Symposium on Logic in Computer Science}} \emph{(\bibinfo{series}{LICS '13})}.
  \bibinfo{publisher}{IEEE Computer Society}, \bibinfo{address}{USA},
  \bibinfo{pages}{530–539}.
\newblock
\showISBNx{9780769550206}
\href{https://doi.org/10.1109/LICS.2013.60}{doi:\nolinkurl{10.1109/LICS.2013.60}}


\bibitem[Damato et~al\mbox{.}(2025)]%
        {damato_formalising_2025}
\bibfield{author}{\bibinfo{person}{Stefania Damato}, \bibinfo{person}{Thorsten
  Altenkirch}, {and} \bibinfo{person}{Axel Ljungström}.}
  \bibinfo{year}{2025}\natexlab{}.
\newblock \bibinfo{title}{Formalising Inductive and Coinductive Containers}.
\newblock
\showeprint[arxiv]{2409.02603 [cs.LO]}
\href{https://doi.org/10.48550/arXiv.2409.02603}{doi:\nolinkurl{10.48550/arXiv.2409.02603}}


\bibitem[Delaware et~al\mbox{.}(2011)]%
        {Delaware2011}
\bibfield{author}{\bibinfo{person}{Benjamin Delaware}, \bibinfo{person}{William
  Cook}, {and} \bibinfo{person}{Don Batory}.} \bibinfo{year}{2011}\natexlab{}.
\newblock \showarticletitle{Product Lines of Theorems}. In
  \bibinfo{booktitle}{\emph{Proceedings of the 2011 ACM International
  Conference on Object Oriented Programming Systems Languages and
  Applications}} (Portland, Oregon, USA) \emph{(\bibinfo{series}{OOPSLA
  2011})}. \bibinfo{publisher}{ACM}, \bibinfo{address}{New York, NY, USA},
  \bibinfo{pages}{595--608}.
\newblock
\showISBNx{978-1-4503-0940-0}
\href{https://doi.org/10.1145/2048066.2048113}{doi:\nolinkurl{10.1145/2048066.2048113}}


\bibitem[Dybjer(1994)]%
        {dybjer_inductive}
\bibfield{author}{\bibinfo{person}{Peter Dybjer}.}
  \bibinfo{year}{1994}\natexlab{}.
\newblock \showarticletitle{Inductive families}.
\newblock \bibinfo{journal}{\emph{Form. Asp. Comput.}} \bibinfo{volume}{6},
  \bibinfo{number}{4} (\bibinfo{date}{July} \bibinfo{year}{1994}),
  \bibinfo{pages}{440–465}.
\newblock
\showISSN{0934-5043}
\href{https://doi.org/10.1007/BF01211308}{doi:\nolinkurl{10.1007/BF01211308}}


\bibitem[First et~al\mbox{.}(2023)]%
        {First2023}
\bibfield{author}{\bibinfo{person}{Emily First}, \bibinfo{person}{Markus~N.
  Rabe}, \bibinfo{person}{Talia Ringer}, {and} \bibinfo{person}{Yuriy Brun}.}
  \bibinfo{year}{2023}\natexlab{}.
\newblock \showarticletitle{Baldur: {{Whole-Proof Generation}} and {{Repair}}
  with {{Large Language Models}}}. In \bibinfo{booktitle}{\emph{Proceedings of
  the 31st {{ACM Joint European Software Engineering Conference}} and
  {{Symposium}} on the {{Foundations}} of {{Software Engineering}}}}
  \emph{(\bibinfo{series}{{{ESEC}}/{{FSE}} 2023})}.
  \bibinfo{publisher}{Association for Computing Machinery},
  \bibinfo{address}{New York, NY, USA}, \bibinfo{pages}{1229--1241}.
\newblock
\showISBNx{9798400703270}
\href{https://doi.org/10.1145/3611643.3616243}{doi:\nolinkurl{10.1145/3611643.3616243}}


\bibitem[Gim\'{e}nez(1994)]%
        {gimenez}
\bibfield{author}{\bibinfo{person}{Eduardo Gim\'{e}nez}.}
  \bibinfo{year}{1994}\natexlab{}.
\newblock \showarticletitle{Codifying Guarded Definitions with Recursive
  Schemes}. In \bibinfo{booktitle}{\emph{Selected Papers from the International
  Workshop on Types for Proofs and Programs}} \emph{(\bibinfo{series}{TYPES
  '94})}. \bibinfo{publisher}{Springer-Verlag}, \bibinfo{address}{Berlin,
  Heidelberg}, \bibinfo{pages}{39–59}.
\newblock
\showISBNx{3540605797}
\href{https://doi.org/10.1007/3-540-60579-7_3}{doi:\nolinkurl{10.1007/3-540-60579-7_3}}


\bibitem[Gopinathan et~al\mbox{.}(2023)]%
        {gopinathan_mostly_2023}
\bibfield{author}{\bibinfo{person}{Kiran Gopinathan}, \bibinfo{person}{Mayank
  Keoliya}, {and} \bibinfo{person}{Ilya Sergey}.}
  \bibinfo{year}{2023}\natexlab{}.
\newblock \showarticletitle{Mostly Automated Proof Repair for Verified
  Libraries}.
\newblock \bibinfo{journal}{\emph{Proc. ACM Program. Lang.}}
  \bibinfo{volume}{7}, \bibinfo{number}{PLDI}, Article \bibinfo{articleno}{107}
  (\bibinfo{date}{June} \bibinfo{year}{2023}), \bibinfo{numpages}{25}~pages.
\newblock
\href{https://doi.org/10.1145/3591221}{doi:\nolinkurl{10.1145/3591221}}


\bibitem[Gylterud et~al\mbox{.}(2026)]%
        {gylterud_terminal_2026}
\bibfield{author}{\bibinfo{person}{Hakon~Robbestad Gylterud},
  \bibinfo{person}{Elisabeth Stenholm}, {and} \bibinfo{person}{Niccolò
  Veltri}.} \bibinfo{year}{2026}\natexlab{}.
\newblock \showarticletitle{Terminal Coalgebras and Non-wellfounded Sets in
  Homotopy Type Theory}.
\newblock   \bibinfo{volume}{Volume 22, Issue 2} (\bibinfo{year}{2026}),
  \bibinfo{pages}{14325}.
\newblock
\showISSN{1860-5974}
\showeprint[arxiv]{2001.06696 [math.LO]}
\href{https://doi.org/10.46298/lmcs-22(2:35)2026}{doi:\nolinkurl{10.46298/lmcs-22(2:35)2026}}


\bibitem[Hagino(1987)]%
        {hagino}
\bibfield{author}{\bibinfo{person}{Tatsuya Hagino}.}
  \bibinfo{year}{1987}\natexlab{}.
\newblock \showarticletitle{A Typed Lambda Calculus with Categorical Type
  Constructors}. In \bibinfo{booktitle}{\emph{Category Theory and Computer
  Science}}. \bibinfo{publisher}{Springer-Verlag}, \bibinfo{address}{Berlin,
  Heidelberg}, \bibinfo{pages}{140–157}.
\newblock
\showISBNx{3540185089}
\href{https://doi.org/10.1007/3-540-18508-9_24}{doi:\nolinkurl{10.1007/3-540-18508-9_24}}


\bibitem[Huffman and Kun{\v{c}}ar(2013)]%
        {isabelletransfer}
\bibfield{author}{\bibinfo{person}{Brian Huffman} {and}
  \bibinfo{person}{Ond{\v{r}}ej Kun{\v{c}}ar}.}
  \bibinfo{year}{2013}\natexlab{}.
\newblock \showarticletitle{Lifting and Transfer: A Modular Design for
  Quotients in Isabelle/HOL}. In \bibinfo{booktitle}{\emph{Certified Programs
  and Proofs}}, \bibfield{editor}{\bibinfo{person}{Georges Gonthier} {and}
  \bibinfo{person}{Michael Norrish}} (Eds.). \bibinfo{publisher}{Springer
  International Publishing}, \bibinfo{address}{Cham},
  \bibinfo{pages}{131--146}.
\newblock
\showISBNx{978-3-319-03545-1}
\href{https://doi.org/10.1007/978-3-319-03545-1_9}{doi:\nolinkurl{10.1007/978-3-319-03545-1_9}}


\bibitem[Jacobs and Rutten(2011)]%
        {Jacobs_Rutten_2011}
\bibfield{author}{\bibinfo{person}{Bart Jacobs} {and} \bibinfo{person}{Jan
  Rutten}.} \bibinfo{year}{2011}\natexlab{}.
\newblock \bibinfo{booktitle}{\emph{An introduction to (co)algebra and
  (co)induction}}.
\newblock \bibinfo{publisher}{Cambridge University Press},
  \bibinfo{pages}{38–99}.
\newblock
\href{https://doi.org/10.1017/CBO9780511792588.003}{doi:\nolinkurl{10.1017/CBO9780511792588.003}}


\bibitem[Jana et~al\mbox{.}(2026)]%
        {jana2026proofbridge}
\bibfield{author}{\bibinfo{person}{Prithwish Jana}, \bibinfo{person}{Kaan
  Kale}, \bibinfo{person}{Ahmet~Ege Tanriverdi}, \bibinfo{person}{Cruise Song},
  \bibinfo{person}{Sriram Vishwanath}, {and} \bibinfo{person}{Vijay Ganesh}.}
  \bibinfo{year}{2026}\natexlab{}.
\newblock \bibinfo{title}{ProofBridge: Auto-Formalization of Natural Language
  Proofs in Lean via Joint Embeddings}.
\newblock
\showeprint[arxiv]{2510.15681}~[cs.LO]
\href{https://doi.org/10.48550/arXiv.2510.15681}{doi:\nolinkurl{10.48550/arXiv.2510.15681}}


\bibitem[KO and GIBBONS(2016)]%
        {ko2016}
\bibfield{author}{\bibinfo{person}{HSIANG-SHANG KO} {and}
  \bibinfo{person}{JEREMY GIBBONS}.} \bibinfo{year}{2016}\natexlab{}.
\newblock \showarticletitle{Programming with ornaments}.
\newblock \bibinfo{journal}{\emph{Journal of Functional Programming}}
  \bibinfo{volume}{27} (\bibinfo{year}{2016}), \bibinfo{pages}{e2}.
\newblock
\href{https://doi.org/10.1017/S0956796816000307}{doi:\nolinkurl{10.1017/S0956796816000307}}


\bibitem[KOZEN and SILVA(2017)]%
        {KOZEN_SILVA_2017}
\bibfield{author}{\bibinfo{person}{DEXTER KOZEN} {and}
  \bibinfo{person}{ALEXANDRA SILVA}.} \bibinfo{year}{2017}\natexlab{}.
\newblock \showarticletitle{Practical coinduction}.
\newblock \bibinfo{journal}{\emph{Mathematical Structures in Computer Science}}
  \bibinfo{volume}{27}, \bibinfo{number}{7} (\bibinfo{year}{2017}),
  \bibinfo{pages}{1132–1152}.
\newblock
\href{https://doi.org/10.1017/S0960129515000493}{doi:\nolinkurl{10.1017/S0960129515000493}}


\bibitem[Laarakker et~al\mbox{.}(2026)]%
        {laarakker_constructing_2026}
\bibfield{author}{\bibinfo{person}{Bastiaan Laarakker},
  \bibinfo{person}{Daniël Otten}, {and} \bibinfo{person}{Benno van~den Berg}.}
  \bibinfo{year}{2026}\natexlab{}.
\newblock \bibinfo{title}{Constructing (Co)inductive Types via Large Sizes}.
\newblock
\showeprint[arxiv]{2602.18921 [cs.LO]}
\href{https://doi.org/10.48550/arXiv.2602.18921}{doi:\nolinkurl{10.48550/arXiv.2602.18921}}


\bibitem[Lu et~al\mbox{.}(2024)]%
        {palm}
\bibfield{author}{\bibinfo{person}{Minghai Lu}, \bibinfo{person}{Benjamin
  Delaware}, {and} \bibinfo{person}{Tianyi Zhang}.}
  \bibinfo{year}{2024}\natexlab{}.
\newblock \showarticletitle{Proof Automation with Large Language Models}. In
  \bibinfo{booktitle}{\emph{Proceedings of the 39th IEEE/ACM International
  Conference on Automated Software Engineering}} (Sacramento, CA, USA)
  \emph{(\bibinfo{series}{ASE '24})}. \bibinfo{publisher}{Association for
  Computing Machinery}, \bibinfo{address}{New York, NY, USA},
  \bibinfo{pages}{1509–1520}.
\newblock
\showISBNx{9798400712487}
\href{https://doi.org/10.1145/3691620.3695521}{doi:\nolinkurl{10.1145/3691620.3695521}}


\bibitem[Masci and Dutle(2022)]%
        {proofmate}
\bibfield{author}{\bibinfo{person}{Paolo Masci} {and} \bibinfo{person}{Aaron
  Dutle}.} \bibinfo{year}{2022}\natexlab{}.
\newblock \showarticletitle{Proof Mate: An Interactive Proof Helper for PVS
  (Tool Paper)}. In \bibinfo{booktitle}{\emph{NASA Formal Methods: 14th
  International Symposium, NFM 2022, Pasadena, CA, USA, May 24–27, 2022,
  Proceedings}} (Pasadena, CA, USA). \bibinfo{publisher}{Springer-Verlag},
  \bibinfo{address}{Berlin, Heidelberg}, \bibinfo{pages}{809–815}.
\newblock
\showISBNx{978-3-031-06772-3}
\href{https://doi.org/10.1007/978-3-031-06773-0_44}{doi:\nolinkurl{10.1007/978-3-031-06773-0_44}}


\bibitem[McBride(2011)]%
        {ornaments2011}
\bibfield{author}{\bibinfo{person}{Conor McBride}.}
  \bibinfo{year}{2011}\natexlab{}.
\newblock \bibinfo{title}{Ornamental algebras, algebraic ornaments}.
\newblock
\urldef\tempurl%
\url{https://plv.mpi-sws.org/plerg/papers/mcbride-ornaments-2up.pdf}
\showURL{%
\tempurl}


\bibitem[Monperrus(2018)]%
        {Monperrus2015}
\bibfield{author}{\bibinfo{person}{Martin Monperrus}.}
  \bibinfo{year}{2018}\natexlab{}.
\newblock \showarticletitle{Automatic Software Repair: A Bibliography}.
\newblock \bibinfo{journal}{\emph{ACM Comput. Surv.}} \bibinfo{volume}{51},
  \bibinfo{number}{1}, Article \bibinfo{articleno}{17} (\bibinfo{date}{Jan.}
  \bibinfo{year}{2018}), \bibinfo{numpages}{24}~pages.
\newblock
\showISSN{0360-0300}
\href{https://doi.org/10.1145/3105906}{doi:\nolinkurl{10.1145/3105906}}


\bibitem[Moura and Ullrich(2021)]%
        {lean}
\bibfield{author}{\bibinfo{person}{Leonardo~de Moura} {and}
  \bibinfo{person}{Sebastian Ullrich}.} \bibinfo{year}{2021}\natexlab{}.
\newblock \showarticletitle{The Lean 4 Theorem Prover and Programming
  Language}. In \bibinfo{booktitle}{\emph{Automated Deduction – CADE 28: 28th
  International Conference on Automated Deduction, Virtual Event, July 12–15,
  2021, Proceedings}}. \bibinfo{publisher}{Springer-Verlag},
  \bibinfo{address}{Berlin, Heidelberg}, \bibinfo{pages}{625–635}.
\newblock
\showISBNx{978-3-030-79875-8}
\href{https://doi.org/10.1007/978-3-030-79876-5_37}{doi:\nolinkurl{10.1007/978-3-030-79876-5_37}}


\bibitem[Nilizadeh et~al\mbox{.}(2023)]%
        {arr-repair}
\bibfield{author}{\bibinfo{person}{Amirfarhad Nilizadeh},
  \bibinfo{person}{Gary~T. Leavens}, {and} \bibinfo{person}{David~R. Cok}.}
  \bibinfo{year}{2023}\natexlab{}.
\newblock \showarticletitle{Automated Reasoning Repair}. In
  \bibinfo{booktitle}{\emph{Proceedings of the 24th ACM International Workshop
  on Formal Techniques for Java-like Programs}} (Berlin, Germany)
  \emph{(\bibinfo{series}{FTfJP '22})}. \bibinfo{publisher}{Association for
  Computing Machinery}, \bibinfo{address}{New York, NY, USA},
  \bibinfo{pages}{11–14}.
\newblock
\showISBNx{9798400707841}
\href{https://doi.org/10.1145/3611096.3611099}{doi:\nolinkurl{10.1145/3611096.3611099}}


\bibitem[Norell(2009)]%
        {agda}
\bibfield{author}{\bibinfo{person}{Ulf Norell}.}
  \bibinfo{year}{2009}\natexlab{}.
\newblock \showarticletitle{Dependently typed programming in Agda}. In
  \bibinfo{booktitle}{\emph{Proceedings of the 4th International Workshop on
  Types in Language Design and Implementation}} (Savannah, GA, USA)
  \emph{(\bibinfo{series}{TLDI '09})}. \bibinfo{publisher}{Association for
  Computing Machinery}, \bibinfo{address}{New York, NY, USA},
  \bibinfo{pages}{1–2}.
\newblock
\showISBNx{9781605584201}
\href{https://doi.org/10.1145/1481861.1481862}{doi:\nolinkurl{10.1145/1481861.1481862}}


\bibitem[Ospanov et~al\mbox{.}(2025)]%
        {ospanov2026apollo}
\bibfield{author}{\bibinfo{person}{Azim Ospanov}, \bibinfo{person}{Farzan
  Farnia}, {and} \bibinfo{person}{Roozbeh Yousefzadeh}.}
  \bibinfo{year}{2025}\natexlab{}.
\newblock \showarticletitle{{APOLLO}: Automated {LLM} and Lean Collaboration
  for Advanced Formal Reasoning}. In \bibinfo{booktitle}{\emph{The Thirty-ninth
  Annual Conference on Neural Information Processing Systems}}.
\newblock
\href{https://doi.org/10.52202/085713-1389}{doi:\nolinkurl{10.52202/085713-1389}}


\bibitem[Patel et~al\mbox{.}(2026)]%
        {patel2026proofrepairbench}
\bibfield{author}{\bibinfo{person}{Manooshree Patel}, \bibinfo{person}{Bartosz
  Piotrowski}, \bibinfo{person}{Leopold Haller}, {and}
  \bibinfo{person}{Hugh~James Leather}.} \bibinfo{year}{2026}\natexlab{}.
\newblock \showarticletitle{ProofRepairBench: Exploring Proof Repair in Lean}.
  In \bibinfo{booktitle}{\emph{ICLR 2026 Workshop: VerifAI-2: The Second
  Workshop on AI Verification in the Wild}}.
\newblock
\urldef\tempurl%
\url{https://openreview.net/forum?id=6SwWVNwEJK}
\showURL{%
\tempurl}


\bibitem[{Plante III} and Paul(2026)]%
        {plante2026artifact}
\bibfield{author}{\bibinfo{person}{Everett {Plante III}} {and}
  \bibinfo{person}{Eric Paul}.} \bibinfo{year}{2026}\natexlab{}.
\newblock \bibinfo{booktitle}{\emph{Artifact Agda formalization for ``Proof
  Repair and (Co)algebra Equivalences''}}.
\newblock
\href{https://doi.org/10.5281/zenodo.23682846}{doi:\nolinkurl{10.5281/zenodo.23682846}}


\bibitem[R and Patawar(2026)]%
        {r2026lampleanbasedagenticframework}
\bibfield{author}{\bibinfo{person}{Santhana~Srinivasan R} {and}
  \bibinfo{person}{Maithilee Patawar}.} \bibinfo{year}{2026}\natexlab{}.
\newblock \bibinfo{title}{LAMP: Lean-based Agentic framework with MCP and Proof
  Repair}.
\newblock
\showeprint[arxiv]{2606.28841}~[cs.LO]
\href{https://doi.org/10.48550/arXiv.2606.28841}{doi:\nolinkurl{10.48550/arXiv.2606.28841}}


\bibitem[Reichel(2024)]%
        {reichel2024neural}
\bibfield{author}{\bibinfo{person}{Thomas Reichel}.}
  \bibinfo{year}{2024}\natexlab{}.
\newblock \emph{\bibinfo{title}{Neural approaches to theorem search \& proof
  repair}}.
\newblock \bibinfo{thesistype}{Master's\ thesis}. \bibinfo{school}{University
  of Illinois at Urbana-Champaign}.
\newblock


\bibitem[Reichel et~al\mbox{.}(2023)]%
        {reichel2023}
\bibfield{author}{\bibinfo{person}{Tom Reichel}, \bibinfo{person}{R.~Wesley
  Henderson}, \bibinfo{person}{Andrew Touchet}, \bibinfo{person}{Andrew
  Gardner}, {and} \bibinfo{person}{Talia Ringer}.}
  \bibinfo{year}{2023}\natexlab{}.
\newblock \showarticletitle{{Proof Repair Infrastructure for Supervised Models:
  Building a Large Proof Repair Dataset}}. In \bibinfo{booktitle}{\emph{14th
  International Conference on Interactive Theorem Proving (ITP 2023)}}
  \emph{(\bibinfo{series}{Leibniz International Proceedings in Informatics
  (LIPIcs)}, Vol.~\bibinfo{volume}{268})},
  \bibfield{editor}{\bibinfo{person}{Adam Naumowicz} {and}
  \bibinfo{person}{Ren\'{e} Thiemann}} (Eds.). \bibinfo{publisher}{Schloss
  Dagstuhl -- Leibniz-Zentrum f{\"u}r Informatik}, \bibinfo{address}{Dagstuhl,
  Germany}, \bibinfo{pages}{26:1--26:20}.
\newblock
\showISBNx{978-3-95977-284-6}
\showISSN{1868-8969}
\href{https://doi.org/10.4230/LIPIcs.ITP.2023.26}{doi:\nolinkurl{10.4230/LIPIcs.ITP.2023.26}}


\bibitem[Ringer(2021)]%
        {ringer_proof_2021}
\bibfield{author}{\bibinfo{person}{Talia Ringer}.}
  \bibinfo{year}{2021}\natexlab{}.
\newblock \emph{\bibinfo{title}{Proof Repair}}.
\newblock \bibinfo{thesistype}{Ph.\,D. Dissertation}.
  \bibinfo{school}{University of Washington}.
\newblock
\urldef\tempurl%
\url{https://digital.lib.washington.edu/researchworks/handle/1773/47429}
\showURL{%
\tempurl}


\bibitem[Ringer et~al\mbox{.}(2019a)]%
        {qedatlarge}
\bibfield{author}{\bibinfo{person}{Talia Ringer}, \bibinfo{person}{Karl
  Palmskog}, \bibinfo{person}{Ilya Sergey}, \bibinfo{person}{Milos Gligoric},
  {and} \bibinfo{person}{Zachary Tatlock}.} \bibinfo{year}{2019}\natexlab{a}.
\newblock \showarticletitle{QED at Large: A Survey of Engineering of Formally
  Verified Software}.
\newblock \bibinfo{journal}{\emph{Found. Trends Program. Lang.}}
  \bibinfo{volume}{5}, \bibinfo{number}{2–3} (\bibinfo{date}{Sept.}
  \bibinfo{year}{2019}), \bibinfo{pages}{102–281}.
\newblock
\showISSN{2325-1107}
\href{https://doi.org/10.1561/2500000045}{doi:\nolinkurl{10.1561/2500000045}}


\bibitem[Ringer et~al\mbox{.}(2021)]%
        {ringer2021pldi}
\bibfield{author}{\bibinfo{person}{Talia Ringer}, \bibinfo{person}{RanDair
  Porter}, \bibinfo{person}{Nathaniel Yazdani}, \bibinfo{person}{John Leo},
  {and} \bibinfo{person}{Dan Grossman}.} \bibinfo{year}{2021}\natexlab{}.
\newblock \showarticletitle{Proof repair across type equivalences}. In
  \bibinfo{booktitle}{\emph{Proceedings of the 42nd ACM SIGPLAN International
  Conference on Programming Language Design and Implementation}} (Virtual,
  Canada) \emph{(\bibinfo{series}{PLDI 2021})}. \bibinfo{publisher}{Association
  for Computing Machinery}, \bibinfo{address}{New York, NY, USA},
  \bibinfo{pages}{112–127}.
\newblock
\showISBNx{9781450383912}
\href{https://doi.org/10.1145/3453483.3454033}{doi:\nolinkurl{10.1145/3453483.3454033}}


\bibitem[Ringer et~al\mbox{.}(2018)]%
        {ringer2018cpp}
\bibfield{author}{\bibinfo{person}{Talia Ringer}, \bibinfo{person}{Nathaniel
  Yazdani}, \bibinfo{person}{John Leo}, {and} \bibinfo{person}{Dan Grossman}.}
  \bibinfo{year}{2018}\natexlab{}.
\newblock \showarticletitle{Adapting Proof Automation to Adapt Proofs}. In
  \bibinfo{booktitle}{\emph{Proceedings of the 7th ACM SIGPLAN International
  Conference on Certified Programs and Proofs}} (Los Angeles, CA, USA)
  \emph{(\bibinfo{series}{CPP 2018})}. \bibinfo{publisher}{Association for
  Computing Machinery}, \bibinfo{address}{New York, NY, USA},
  \bibinfo{pages}{115–129}.
\newblock
\showISBNx{9781450355865}
\href{https://doi.org/10.1145/3167094}{doi:\nolinkurl{10.1145/3167094}}


\bibitem[Ringer et~al\mbox{.}(2019b)]%
        {ringer2019orn}
\bibfield{author}{\bibinfo{person}{Talia Ringer}, \bibinfo{person}{Nathaniel
  Yazdani}, \bibinfo{person}{John Leo}, {and} \bibinfo{person}{Dan Grossman}.}
  \bibinfo{year}{2019}\natexlab{b}.
\newblock \showarticletitle{{Ornaments for Proof Reuse in Coq}}. In
  \bibinfo{booktitle}{\emph{10th International Conference on Interactive
  Theorem Proving (ITP 2019)}} \emph{(\bibinfo{series}{Leibniz International
  Proceedings in Informatics (LIPIcs)}, Vol.~\bibinfo{volume}{141})},
  \bibfield{editor}{\bibinfo{person}{John Harrison}, \bibinfo{person}{John
  O'Leary}, {and} \bibinfo{person}{Andrew Tolmach}} (Eds.).
  \bibinfo{publisher}{Schloss Dagstuhl -- Leibniz-Zentrum f{\"u}r Informatik},
  \bibinfo{address}{Dagstuhl, Germany}, \bibinfo{pages}{26:1--26:19}.
\newblock
\showISBNx{978-3-95977-122-1}
\showISSN{1868-8969}
\href{https://doi.org/10.4230/LIPIcs.ITP.2019.26}{doi:\nolinkurl{10.4230/LIPIcs.ITP.2019.26}}


\bibitem[Robert(2018)]%
        {robert2018}
\bibfield{author}{\bibinfo{person}{Valentin Robert}.}
  \bibinfo{year}{2018}\natexlab{}.
\newblock \emph{\bibinfo{title}{Front-end tooling for building and maintaining
  dependently-typed functional programs}}.
\newblock \bibinfo{thesistype}{Ph.\,D. Dissertation}. \bibinfo{school}{UC San
  Diego}.
\newblock


\bibitem[Rutten(2000)]%
        {rutten_universal_2000}
\bibfield{author}{\bibinfo{person}{J.~J. M.~M. Rutten}.}
  \bibinfo{year}{2000}\natexlab{}.
\newblock \showarticletitle{Universal coalgebra: a theory of systems}.
\newblock  \bibinfo{volume}{249}, \bibinfo{number}{1} (\bibinfo{year}{2000}),
  \bibinfo{pages}{3--80}.
\newblock
\showISSN{0304-3975}
\href{https://doi.org/10.1016/S0304-3975(00)00056-6}{doi:\nolinkurl{10.1016/S0304-3975(00)00056-6}}


\bibitem[Tabareau et~al\mbox{.}(2018)]%
        {equivalencesforfree}
\bibfield{author}{\bibinfo{person}{Nicolas Tabareau}, \bibinfo{person}{\'{E}ric
  Tanter}, {and} \bibinfo{person}{Matthieu Sozeau}.}
  \bibinfo{year}{2018}\natexlab{}.
\newblock \showarticletitle{Equivalences for Free: Univalent Parametricity for
  Effective Transport}.
\newblock \bibinfo{journal}{\emph{Proc. ACM Program. Lang.}}
  \bibinfo{volume}{2}, \bibinfo{number}{ICFP}, Article \bibinfo{articleno}{92}
  (\bibinfo{date}{jul} \bibinfo{year}{2018}), \bibinfo{numpages}{29}~pages.
\newblock
\href{https://doi.org/10.1145/3236787}{doi:\nolinkurl{10.1145/3236787}}


\bibitem[Tabareau et~al\mbox{.}(2021)]%
        {marriage}
\bibfield{author}{\bibinfo{person}{Nicolas Tabareau}, \bibinfo{person}{\'{E}ric
  Tanter}, {and} \bibinfo{person}{Matthieu Sozeau}.}
  \bibinfo{year}{2021}\natexlab{}.
\newblock \showarticletitle{The Marriage of Univalence and Parametricity}.
\newblock \bibinfo{journal}{\emph{J. ACM}} \bibinfo{volume}{68},
  \bibinfo{number}{1}, Article \bibinfo{articleno}{5} (\bibinfo{date}{jan}
  \bibinfo{year}{2021}), \bibinfo{numpages}{44}~pages.
\newblock
\showISSN{0004-5411}
\href{https://doi.org/10.1145/3429979}{doi:\nolinkurl{10.1145/3429979}}


\bibitem[Tahat et~al\mbox{.}(2024)]%
        {tahat2024proof}
\bibfield{author}{\bibinfo{person}{Amer Tahat}, \bibinfo{person}{David Hardin},
  \bibinfo{person}{Adam Petz}, {and} \bibinfo{person}{Perry Alexander}.}
  \bibinfo{year}{2024}\natexlab{}.
\newblock \showarticletitle{Proof repair utilizing large language models: A
  case study on the copland remote attestation proofbase}. In
  \bibinfo{booktitle}{\emph{International Conference on Bridging the Gap
  between AI and Reality}}. Springer, \bibinfo{pages}{145--166}.
\newblock
\href{https://doi.org/10.1007/978-3-031-75434-0_10}{doi:\nolinkurl{10.1007/978-3-031-75434-0_10}}


\bibitem[{Univalent Foundations Program}(2013)]%
        {hottbook}
\bibfield{author}{\bibinfo{person}{The {Univalent Foundations Program}}.}
  \bibinfo{year}{2013}\natexlab{}.
\newblock \bibinfo{booktitle}{\emph{Homotopy Type Theory: Univalent Foundations
  of Mathematics}}.
\newblock \bibinfo{publisher}{\url{https://homotopytypetheory.org/book}},
  \bibinfo{address}{Institute for Advanced Study}.
\newblock


\bibitem[Vezzosi et~al\mbox{.}(2019)]%
        {vezzosi_cubical_agda}
\bibfield{author}{\bibinfo{person}{Andrea Vezzosi}, \bibinfo{person}{Anders
  M\"{o}rtberg}, {and} \bibinfo{person}{Andreas Abel}.}
  \bibinfo{year}{2019}\natexlab{}.
\newblock \showarticletitle{Cubical agda: a dependently typed programming
  language with univalence and higher inductive types}.
\newblock \bibinfo{journal}{\emph{Proc. ACM Program. Lang.}}
  \bibinfo{volume}{3}, \bibinfo{number}{ICFP}, Article \bibinfo{articleno}{87}
  (\bibinfo{date}{July} \bibinfo{year}{2019}), \bibinfo{numpages}{29}~pages.
\newblock
\href{https://doi.org/10.1145/3341691}{doi:\nolinkurl{10.1145/3341691}}


\bibitem[Viola et~al\mbox{.}(2025)]%
        {viola2025oopsla}
\bibfield{author}{\bibinfo{person}{Cosmo Viola}, \bibinfo{person}{Max Fan},
  {and} \bibinfo{person}{Talia Ringer}.} \bibinfo{year}{2025}\natexlab{}.
\newblock \showarticletitle{Proof Repair across Quotient Type Equivalences}.
\newblock \bibinfo{journal}{\emph{Proc. ACM Program. Lang.}}
  \bibinfo{volume}{9}, \bibinfo{number}{OOPSLA2}, Article
  \bibinfo{articleno}{386} (\bibinfo{date}{Oct.} \bibinfo{year}{2025}),
  \bibinfo{numpages}{26}~pages.
\newblock
\href{https://doi.org/10.1145/3763164}{doi:\nolinkurl{10.1145/3763164}}


\bibitem[Wang et~al\mbox{.}(2026)]%
        {wang2026learningrepairleanproofs}
\bibfield{author}{\bibinfo{person}{Evan Wang}, \bibinfo{person}{Simon Chess},
  \bibinfo{person}{Daniel Lee}, \bibinfo{person}{Siyuan Ge},
  \bibinfo{person}{Ajit Mallavarapu}, \bibinfo{person}{Jarod Alper}, {and}
  \bibinfo{person}{Vasily Ilin}.} \bibinfo{year}{2026}\natexlab{}.
\newblock \bibinfo{title}{Learning to Repair Lean Proofs from Compiler
  Feedback}.
\newblock
\showeprint[arxiv]{2602.02990}~[cs.LG]
\href{https://doi.org/10.48550/arXiv.2602.02990}{doi:\nolinkurl{10.48550/arXiv.2602.02990}}


\bibitem[Woos et~al\mbox{.}(2016)]%
        {Woos2016}
\bibfield{author}{\bibinfo{person}{Doug Woos}, \bibinfo{person}{James~R.
  Wilcox}, \bibinfo{person}{Steve Anton}, \bibinfo{person}{Zachary Tatlock},
  \bibinfo{person}{Michael~D. Ernst}, {and} \bibinfo{person}{Thomas Anderson}.}
  \bibinfo{year}{2016}\natexlab{}.
\newblock \showarticletitle{Planning for Change in a Formal Verification of the
  Raft Consensus Protocol}. In \bibinfo{booktitle}{\emph{Proceedings of the 5th
  ACM SIGPLAN Conference on Certified Programs and Proofs}} (St. Petersburg,
  FL, USA) \emph{(\bibinfo{series}{CPP 2016})}. \bibinfo{publisher}{ACM},
  \bibinfo{address}{New York, NY, USA}, \bibinfo{pages}{154--165}.
\newblock
\showISBNx{978-1-4503-4127-1}
\href{https://doi.org/10.1145/2854065.2854081}{doi:\nolinkurl{10.1145/2854065.2854081}}


\bibitem[Yang et~al\mbox{.}(2026)]%
        {yang2026exverus}
\bibfield{author}{\bibinfo{person}{Jun Yang}, \bibinfo{person}{Yuechun Sun},
  \bibinfo{person}{Yi Wu}, \bibinfo{person}{Rodrigo Caridad},
  \bibinfo{person}{Yongwei Yuan}, \bibinfo{person}{Jianan Yao},
  \bibinfo{person}{Shan Lu}, {and} \bibinfo{person}{Kexin Pei}.}
  \bibinfo{year}{2026}\natexlab{}.
\newblock \showarticletitle{ExVerus: Verus Proof Repair via Counterexample
  Reasoning}. In \bibinfo{booktitle}{\emph{Forty-third International Conference
  on Machine Learning}}.
\newblock
\href{https://doi.org/10.48550/arXiv.2603.25810}{doi:\nolinkurl{10.48550/arXiv.2603.25810}}


\end{thebibliography}

\end{document}